%% file: main.tex
\documentclass[sigconf, 10pt]{acmart}

\AtBeginDocument{%
  }

\usepackage{balance} 
\usepackage{etoolbox}
\usepackage{stfloats}
\usepackage{placeins}
\usepackage{float}
\usepackage{color}
\usepackage{subfigure}
\usepackage{subcaption}
\usepackage{amsmath}
\usepackage{array}
\usepackage{tabularx}
\usepackage{graphicx}
\usepackage[table]{xcolor}
\usepackage{hyperref}
\usepackage{booktabs}
\usepackage{makecell}
\usepackage{enumitem}
\usepackage{caption}
\usepackage{multirow}
\usepackage{natbib}
\usepackage{tikz}
\usepackage{adjustbox}
\usepackage[linesnumbered, boxed,ruled,vlined]{algorithm2e}
\usepackage[shortcuts]{extdash}
\usepackage{amsthm}
\usepackage{balance}

\definecolor{slotcolor}{RGB}{180,60,60}
\usepackage{listings}
\usepackage{xcolor}
\lstdefinelanguage{STP}{
    keywords={define, function, input, return, if, else, while, for, each, in, iterate, over, match, break, continue, record, set, as, run, execute, with, append, sort, merge, query, add, to},
    keywordstyle=\color{blue}\bfseries,
    ndkeywords={true, false, null},
    ndkeywordstyle=\color{purple}\bfseries,
    identifierstyle=\color{black},
    sensitive=false,
    comment=[l]{\#},
    commentstyle=\color{gray}\itshape,
    stringstyle=\color{red},
    morestring=[b]",
    morestring=[b]',
}

\lstdefinestyle{prompt}{
    language={},
    basicstyle=\ttfamily\footnotesize,
    frame=single,
    breaklines=true,
    columns=fullflexible,
    keepspaces=true,
    numbers=none,
    commentstyle=\color{gray}\itshape,
    stringstyle=\color{red!70!black},
    mathescape=false,
    captionpos=b
}

\newcommand{\systemname}{MANA\xspace}
\newcommand{\systemnames}{MANA's\xspace}
\newcommand{\datasetname}{MANAZoo\xspace}

\newcommand{\ie}{\textit{i}.\textit{e}.\xspace}
\newcommand{\eg}{\textit{e}.\textit{g}.\xspace}

\newcommand*\circlednum[1]{%
  \tikz[baseline=(char.base)]{
    \node[shape=circle,fill=black,text=white,inner sep=1pt] (char) {#1};
  }%
}

\newcommand{\red}[1]{\iffalse #1 \fi}
\newcommand{\blue}[1]{\textcolor{black}{#1}}

\patchcmd{\abstractname}{Abstract}{ABSTRACT}{}{}
\begin{document}

\title{\systemname: Towards Efficient Mobile Ad Detection via Multimodal Agentic UI Navigation}
\author{Yizhe Zhao$^{\ddag}$, Yongjian Fu$^{\ddag \diamond}$, Zihao Feng$^{\dag}$, Hao Pan$^{\S}$, Yongheng Deng$^{\ddag}$, Yaoxue Zhang$^{\ddag}$, Ju Ren$^{\ddag \P \diamond}$}
\thanks{$^{\diamond}$ Corresponding author.}

 \affiliation{$^{\ddag}$Department of Computer Science and Technology, Tsinghua University \country{China}, \\
 $^{\dag}$University of Southern California, USA, $^{\S}$ Shanghai Jiao Tong University, China, \\
 $^{\P}$ State Key Laboratory of Internet Architecture, Tsinghua University, China, \\
$^{\ddag}$zhao-yz24@mails.tsinghua.edu.cn, 
$^{\ddag}$\{fuyongjian,dyh2024,zhangyx,renju\}@tsinghua.edu.cn,
$^{\dag}$zihaofen@usc.edu, 
$^{\S}$panh09@sjtu.edu.cn
}
\email{}



\settopmatter{printacmref=false}
\settopmatter{printfolios=false}
\renewcommand\footnotetextcopyrightpermission[1]{} 

\input{chapter/abstract}



\renewcommand{\shortauthors}{Zhao, et al}

\maketitle

\pagestyle{plain}
\input{chapter/introduction}
\input{chapter/background}

\input{chapter/design}
\input{chapter/implementation}
\input{chapter/evaluation}

\input{chapter/related_work}
\input{chapter/discussion}
\input{chapter/conclusion}
\input{chapter/acknowledgement}

\bibliographystyle{ACM-Reference-Format}
\bibliography{reference}
\input{chapter/appendix}
\end{document}

%% file: chapter/abstract.tex
\begin{abstract}
Mobile advertising dominates app monetization but introduces risks ranging from intrusive user experience to malware delivery. 
Existing detection methods rely either on static analysis, which misses runtime behaviors, or on heuristic UI exploration, which struggles with sparse and obfuscated ads. 
In this paper, we present \systemname, the first agentic multimodal reasoning framework for mobile ad detection. 
\systemname integrates static, visual, temporal, and experiential signals into a reasoning-guided navigation strategy that determines not only how to traverse interfaces but also where to focus, enabling efficient and robust exploration. 
We implement and evaluate \systemname on commercial smartphones over $200$ apps, achieving state-of-the-art accuracy and efficiency. 
Compared to baselines, it improves detection accuracy by $30.5\%$–$56.3\%$ and reduces exploration steps by $29.7\%$–$63.3\%$. 
Case studies further demonstrate its ability to uncover obfuscated and malicious ads, underscoring its practicality for mobile ad auditing and its potential for broader runtime UI analysis~(\eg, permission abuse). Code and dataset are available at \url{https://github.com/MANA-2026/MANA}.

\end{abstract}

%% file: chapter/introduction.tex
\section{Introduction}
Mobile advertising has become the economic backbone of the mobile ecosystem, with a global market projected to surpass USD 1 trillion by 2032~\cite{choudhury2025mobile}. 
Recent reports show that ads appear in over 60\% of Google Play apps~\cite{9052472}, most commonly delivered through third-party SDKs such as Google AdMob~\cite{AdMob2024}, Meta Audience Network~\cite{FacebookAudienceNetwork}, and Applovin~\cite{AppLovin2024}.
However, this prevalence introduces substantial risks. 
Ads may disrupt interaction with intrusive pop-ups~\cite{hanbazazh2021pop}, simulate native interface elements through deceptive design~\cite{zha2014impact}, or be weaponized to deliver malware and trigger malicious redirects~\cite{subramani2020push}.
In the current environment, ads are not merely peripheral annoyances but instead constitute a significant attack surface with direct implications for both user security and platform integrity.
Consequently, the ability to automatically and reliably detect in-app ads is no longer a matter of mere convenience, but a critical requirement for vetting apps at scale and securing the mobile ecosystem.

Pioneering efforts in mobile ad detection have progressed along two main directions. 
Static analysis~\cite{lee2019adlib,feldman2014manilyzer,lee2016hybridroid,xinyu2023andetect} inspects manifests, layouts, or bytecode to identify ad SDK signatures. 
While scalable, this approach offers no visibility into when or where ads actually appear at runtime and is easily defeated by obfuscation or dynamic code loading.
To address these limitations, dynamic exploration driven by automated UI navigation has become the dominant paradigm.
Tools such as Monkey~\cite{monkey2023} rely on random event generation, DroidBot~\cite{li2017droidbot} maintains a lightweight state model, and systems like MadDroid~\cite{liu2020maddroid} and ADGPE~\cite{ma2024careful} prioritize UI components based on keywords or class names, while FraudDroid~\cite{dong2018frauddroid} combines exploration with network traffic analysis.
Despite improving runtime visibility, these methods are fundamentally limited by their reliance on narrow, single-source signals (\eg, view metadata).
%
Their navigation strategies remain brittle and inefficient, often cycling through shallow states while failing to uncover ads hidden in deeper, more complex interaction flows.

\begin{figure}[!t]
\hspace*{0.2cm}
\includegraphics[width=0.46\textwidth]{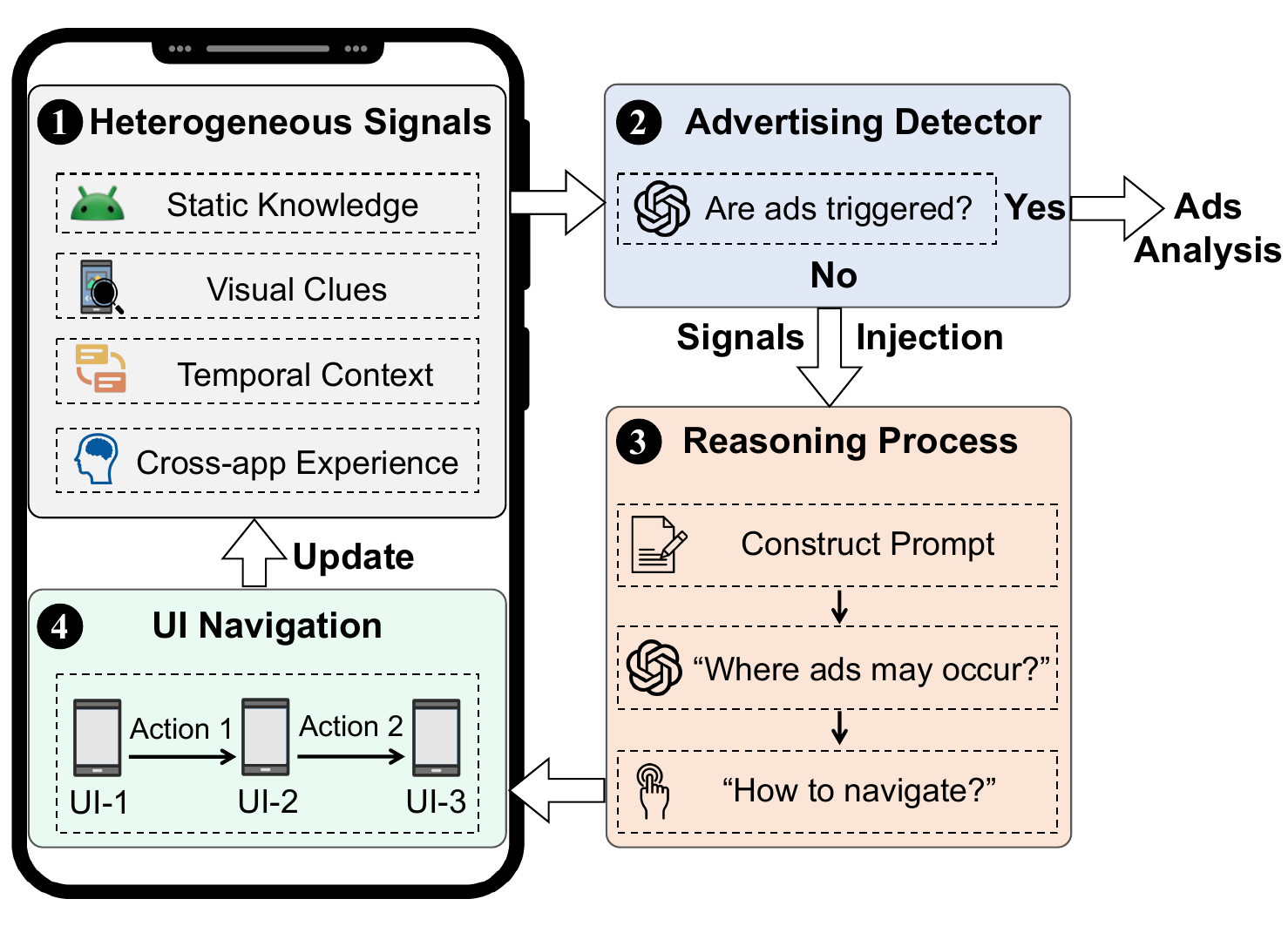}
  \vspace{-0.3cm}
  \caption{\systemname as a multimodal reasoning agent for efficient mobile ad detection, fusing heterogeneous signals to jointly reason about ``where and how to go''.}
  \label{fig:scenario}
\end{figure}
This limitation becomes critical as modern apps adopt evasive patterns that break metadata-dependent heuristics.
UI metadata is often missing, ambiguous, or bypassed, as seen with indistinguishable \texttt{ImageView} elements or in apps (\eg, games) that use full canvas rendering, concealing the entire widget hierarchy.
%
Such practices underscore the need for a paradigm shift: from brittle, single-signal heuristics to a robust reasoning framework that can fuse heterogeneous signals to guide exploration.
Therefore, \textit{how to turn diverse and often incomplete runtime evidence into coherent guidance for efficiently exposing hidden ads remains an open challenge}.

In this paper, we propose \systemname, a \underline{M}ultimodal \underline{A}gentic UI \underline{N}avigation framework for mobile \underline{A}d detection that operationalizes this paradigm shift, as depicted in Figure~\ref{fig:scenario}.
%
\systemname moves beyond simple heuristics to act as a reasoning agent that fuses static footprints, visual cues, temporal context and cross-app experience to steer exploration toward high-yield states.
By leveraging large language and vision\-/language models (LLMs/VLMs)~\cite{openai2023gpt4v,li2024llava}, it interprets ambiguous and dynamically evolving heterogeneous signals at runtime, enabling robust and efficient discovery of ads concealed in diverse interaction flows.
\systemname contributes to agentic UI navigation by shifting the focus from goal-oriented task completion~\cite{liu2024make, zhao2025llm, wen2024autodroid} to adversarial exploratory search.
Generic agents excel at explicit task completion (\eg, ``add item to cart''), where the goal is clear and targets are structurally exposed. 
In contrast, ad detection is an \textit{exploratory search problem}: the targets are hidden, sparsely distributed, and triggered by non-obvious event sequences.
This challenge is not just ``\textit{how to go}'' but also ``\textit{where to go}'' when the destination is unknown.
This necessitates a reasoning framework over heterogeneous signals, where the model synthesizes static footprints with dynamic cues to infer latent UI states, addressing a form of ``detecting the unseen'' that extends beyond standard LLM-based navigation heuristics.




To realize this vision, our design overcomes three fundamental challenges.

First, \textbf{addressing the cold start problem}. 
Automated exploration starts without prior knowledge of ad locations, resulting in inefficient and budget-wasting searches across the vast state space.
To mitigate this, we conduct a large-scale empirical analysis of advertising apps, revealing that ad-hosting structures consistently leave traces in manifests, layouts, and bytecode. 
Building on these insights, we introduce \textit{offline profiling}, which combines static analysis with lightweight probing to derive actionable priors (screen, slot, trigger, network) and to construct an initial User Transition Graph~(UTG). 
These priors seed navigation with structural knowledge, providing principled guidance from the outset.

Second, \textbf{overcoming fragmented observability}. 
%
Ad\-/related signals are scattered across modalities and time.
Metadata may be sparse, canvas-based UIs hide their structure, and visual cues are often subtle.
Furthermore, ads can be time-gated or require specific interaction sequences. 
Our insight is that observability must be reconstructed by fusing modalities and aligning them across time.
We therefore design \textit{multimodal reasoning–guided navigation}, which unifies metadata with visual cues, grounds exploration in temporal and structural context via the UTG and interaction history.
It further leverages LLM/VLM reasoning to interpret incomplete and dynamic signals into a coherent action plan, guiding not only how to navigate the interface but also where to focus in order to surface hidden ads.


Third, \textbf{bridging the cross-app knowledge gap}. 
Ad\-/triggering patterns often repeat across different apps, but traditional explorers learn nothing from past runs, repeatedly rediscovering the same strategies.
%
This failure to generalize leads to redundant effort and poor scalability.
Motivated by the observation that ad-triggering patterns exhibit strong structural and semantic regularities, we propose \textit{memory-driven runtime optimization}.
The system records successful ad-triggered trajectories as ``experiences'', abstracts them into reusable, semantically encoded strategies, and transfers this knowledge to new exploration contexts.
By learning from past successes, the agent accelerates discovery in new apps and makes more informed decisions about where to search and how to proceed.



We implement \systemname on commercial smartphones~(\ie, Honor V20 and Huawei Mate 40 Pro), leveraging heterogeneous signals for agentic multimodal reasoning. 
Evaluation on $200$ apps (\ie, ADGPE and a newly curated \datasetname) shows that \systemname achieves state-of-the-art performance, reaching detection rates of $77.0\%$ with improvements of $30.5\%$–$56.3\%$ over baselines and reducing redundant exploration paths by $29.7\%$–$63.3\%$. 
\systemname further adapts to diverse reasoning models (\ie, GPT-4o mini~\cite{hurst2024gpt}, Llama-4~\cite{meta2024llama4}, Qwen-2.5~\cite{qwen2p5_7b}) and mobile devices, with case studies confirming its ability to uncover obfuscated and policy-violating ads. 
Beyond ad auditing, it generalizes to broader runtime UI analysis tasks such as malware detection, while preserving the same agentic multimodal reasoning workflow.


Our contributions are threefold:
\begin{itemize}[leftmargin=*]
    \item We introduce \systemname, the first multimodal agentic framework for mobile ad detection. It pioneers an exploratory search paradigm that fuses static, visual, temporal, and experiential signals to discover sparse, evasive, and diverse advertising behaviors.  

    \item We design a reasoning-guided UI navigation strategy that integrates heterogeneous and evolving signals to infer hidden ad-related triggers. This enables the agent to determine both where to focus and how to traverse the interface, ensuring efficient and robust exploration.  
    
    
    
    \item We implement and evaluate \systemname on commercial smartphones across two representative datasets. \systemname attains state-of-the-art detection accuracy and efficiency, with case studies demonstrating its effectiveness in uncovering obfuscated and malicious ads. These findings underscore its practicality for mobile ad auditing and broader runtime UI analysis~(\eg, permission abuse). Code and dataset will be released to facilitate future research.
\end{itemize}

%% file: chapter/background.tex
\section{Background and Motivation}
\subsection{The workflow of mobile advertising.} 
\label{sec:adbackground}
Mobile ad in Android apps is primarily enabled through the integration of third-party ad SDKs (\eg, Google AdMob~\cite{AdMob2024}, Meta Audience Network~\cite{andreou2019measuring, andreou2018investigating}).
Once embedded, these SDKs dynamically fetch and render ad creatives from remote ad networks at runtime, supporting diverse formats such as banners, interstitials, native, or rewarded ads~\cite{ma2024careful}.
%
%
When triggered, an ad widget initiates a network request to retrieve content, which may lead to further UI transitions.
For example, redirecting the user to an in-app landing page, the Google Play Store, or external web content~\cite{liu2020maddroid}. 
In \systemname, these behaviors are treated as potential ad triggers.
Modern ad platforms increasingly rely on runtime auctions and user profiling for monetization, but this also brings risks such as intrusive ads, unwanted redirections, and even malware delivery~\cite{venkatadri2018privacy,speicher2018potential}.
%
%
%
%
Diverse ad embedding and runtime delivery strategies call for robust ad detection frameworks that generalize across app categories.
%
We group mobile ads into three types (Figure~\ref{fig:ad-type}): \textit{Embedded ads} blend into the content flow and may resemble functional UI; \textit{Popup ads} are interruptive overlays and are typically easier to detect; \textit{Custom ads} use app-specific/non-standard rendering with few SDK or textual cues, making them the hardest to identify.

%
%
%

\begin{figure}[!t]
\includegraphics[width=0.48\textwidth]{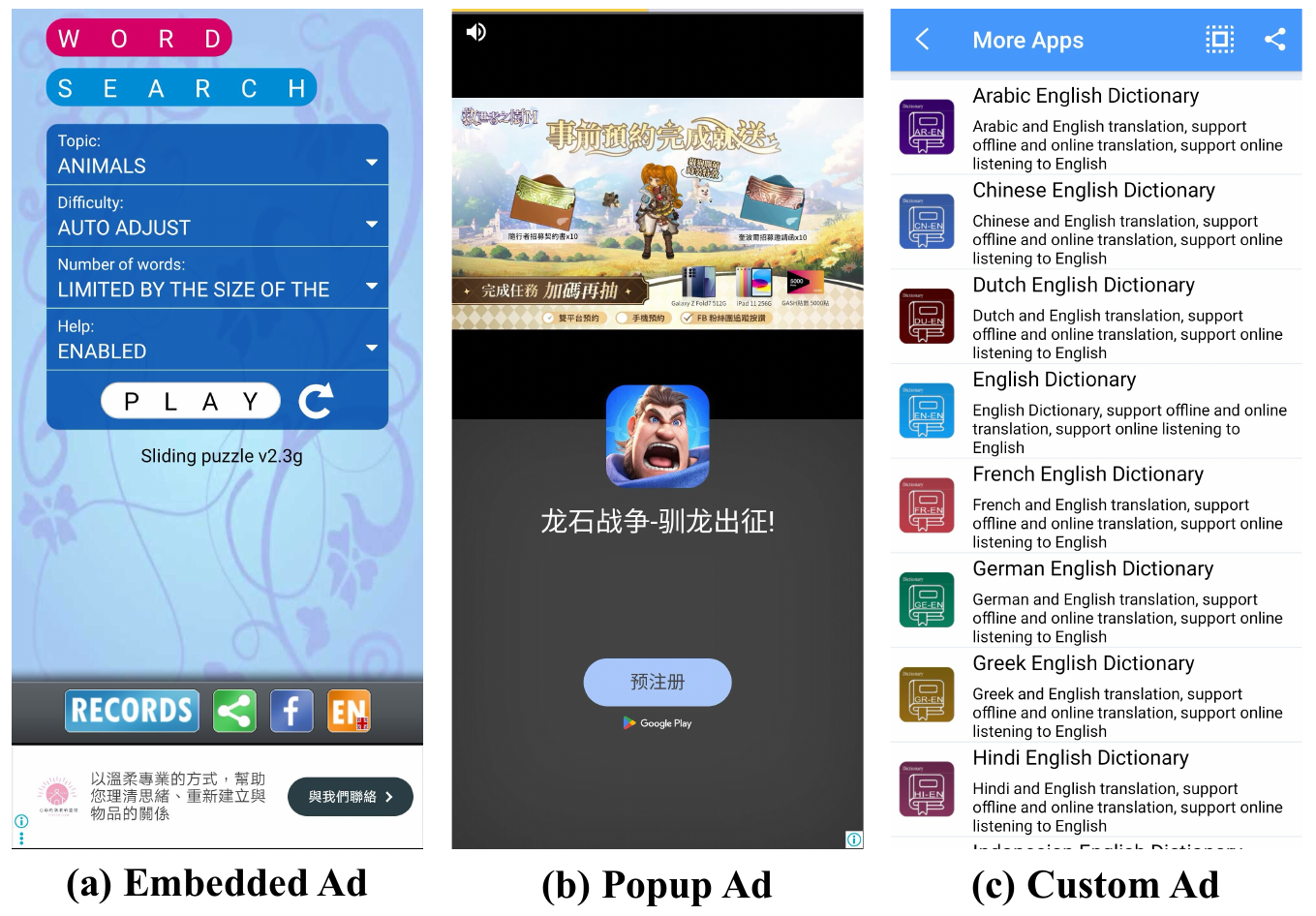}
  \vspace{-0.6cm}
  \caption{Examples of different ad types: (a) integrated within the app interface as a banner; (b) shown as an intrusive interstitial that blocks normal interaction; (c) implemented as a ``More Apps'' list without standard SDK patterns, requiring semantic reasoning to recognize as advertising.}
  \label{fig:ad-type}
  \vspace{-0.2cm}
\end{figure}

\subsection{Challenges in Mobile Ad Detection.}
\label{sec:challenges}
In this section, we reveal that ads in mobile apps exhibit three fundamental properties that complicate detection.

\noindent \textbf{Sparsity and Obscured.}
Ad-triggering events are rare and often depend on specific multi-step, time-gated interaction sequences rather than immediate UI feedback (e.g., “Subscribe” $\rightarrow$ “Watch Video” $\rightarrow$ wait for loading). Large-scale studies confirm this sparsity: automated exploration can spend up to $98.6\%$ of time in repetitive loops without reaching deeper ad-hosting states~\cite{wang2021vet}, and many ads only appear along deep execution paths~\cite{liu2020maddroid}. As a result, blind exploration wastes budget and easily misses hidden ads, motivating reasoning-guided navigation.
%
%
%
%

\begin{figure}[!t]
\includegraphics[width=0.48\textwidth]{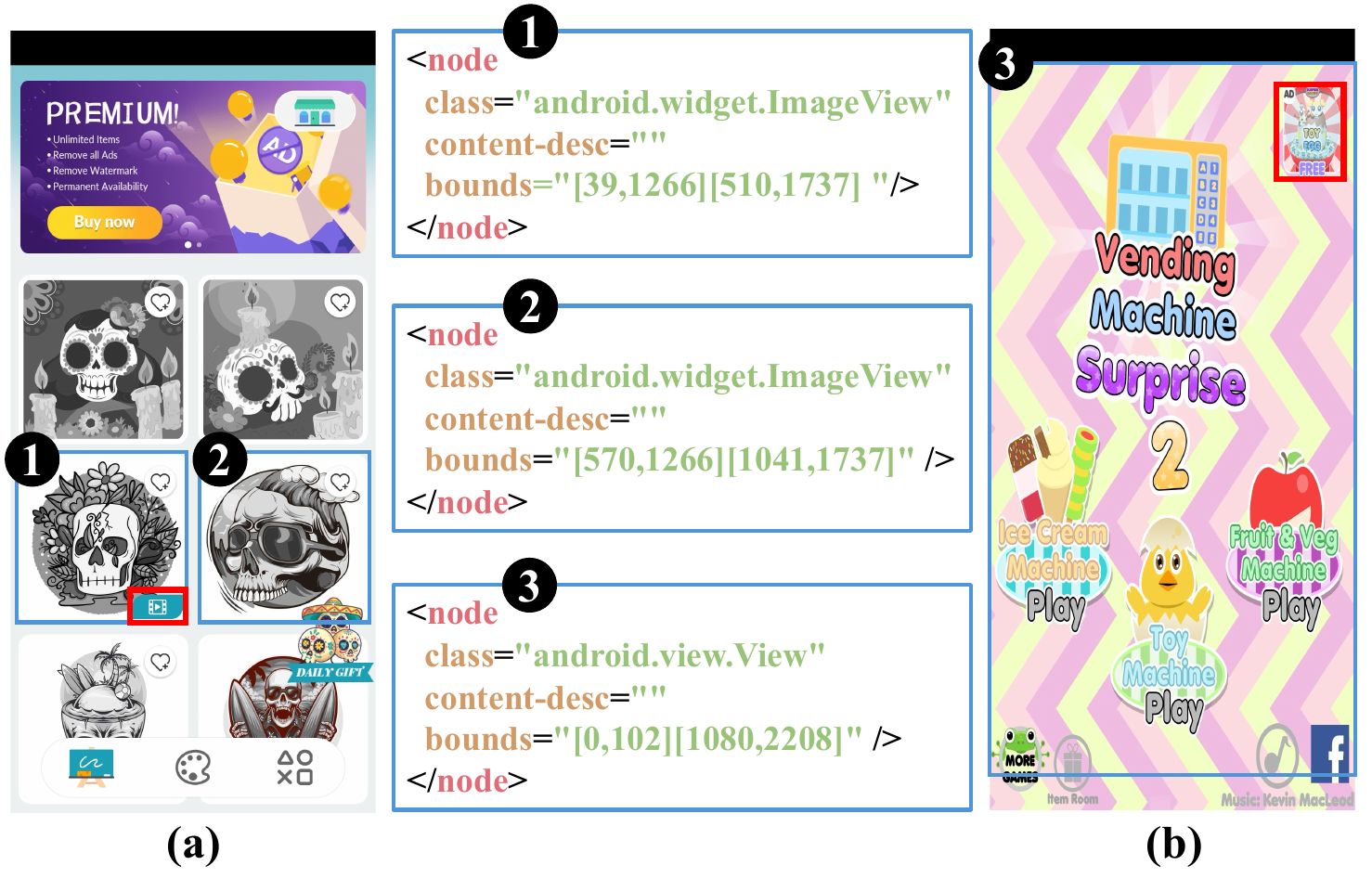}
  \vspace{-0.8cm}
  \caption{Screenshots of two example apps. The middle panels show the view-hierarchy nodes corresponding to the blue-boxed components in the screenshots, while red boxes highlight ad-related elements.}
  \label{fig:preliminary}
\end{figure}

\noindent \textbf{Visual Cues Beyond Metadata.}
Many ad detectors rely on view metadata (e.g., \texttt{class}, \texttt{content-desc})\cite{liu2020maddroid, ma2024careful}, but these fields are often empty, generic, or ambiguous. As Figure\ref{fig:preliminary}(a) shows, two \texttt{ImageView} elements share identical metadata~(\ie, \circlednum{1} and \circlednum{2}), yet only one visually contains a play icon indicating a video ad; metadata-only methods cannot separate them. Prior work also shows developers may deliberately obscure ad cues, such as embedding close buttons into banners or mimicking native UI elements~\cite{mathur2019dark}. Together, these observations motivate vision-based reasoning for robust ad detection.
%
%
%
%
%
%

\noindent \textbf{Heterogeneous Rendering Pipelines.}
Many apps expose structured view hierarchies that can be systematically parsed via the Accessibility Service~\cite{androidAccessibilityService}.
In contrast, an increasing fraction of apps, especially games, adopt canvas-based rendering (\eg, Unity). 
In such cases, the entire interface is painted as a bitmap, leaving no widget\-/level metadata available for accessibility-based inspection~\cite{ye2021empirical, rastogi2016these}. 
As shown in Figure~\ref{fig:preliminary}(b), ad\-/triggering components become invisible to the hierarchy, making screenshot-level analysis the only viable option~(\ie, \circlednum{3}). 
This heterogeneity highlights the need for a unified abstraction that can normalize accessibility\-/based and vision-based evidence into a consistent representation of actionable UI elements.

\subsection{Multimodal Reasoning-Guided UI Navigation for Ad Detection.}
\label{sec:opportunity}
Mobile ads exhibit sparse triggers, subtle visual cues, and diverse rendering pipelines, which impose three methodological requirements: (i) reasoning over sparse and obscured signals, (ii) leveraging visual cues to distinguish ad-specific patterns, and (iii) normalizing heterogeneous rendering mechanisms into a consistent abstraction of actionable UI components. Conventional approaches such as static SDK detectors, metadata-based classifiers, and heuristic exploration strategies cannot satisfy these requirements; they often fail to localize ads, miss critical visual evidence, or waste interaction budget in redundant exploration~\cite{liu2020maddroid, ma2024careful, mathur2019dark, mao2016sapienz}.

LLMs and VLMs enable multimodal UI reasoning: LLMs process textual metadata and context, while VLMs capture visual cues from screenshots. Prior work shows that combining textual, visual, and historical signals prioritizes high-yield paths and outperforms RL or heuristic navigation~\cite{wen2024autodroid, liu2024make, zhao2025llm}, making these models promising for ad-oriented exploration.
However, generic task-driven UI agents assume explicit goals and dense rewards, whereas ads are sparse, deliberately obscured, and often buried deep in interaction paths. Effective ad exploration therefore requires reasoning over temporal and visual evidence to surface ad cues and identify which UI components to interact with to trigger ads. To meet these needs, we build a multimodal structured reasoning agent on LLMs and VLMs for targeted mobile ad exploration, rather than a generic click predictor.

%% file: chapter/design.tex
\section{\systemname Design}
\begin{figure*}[h]
  \includegraphics[width=0.99\linewidth]{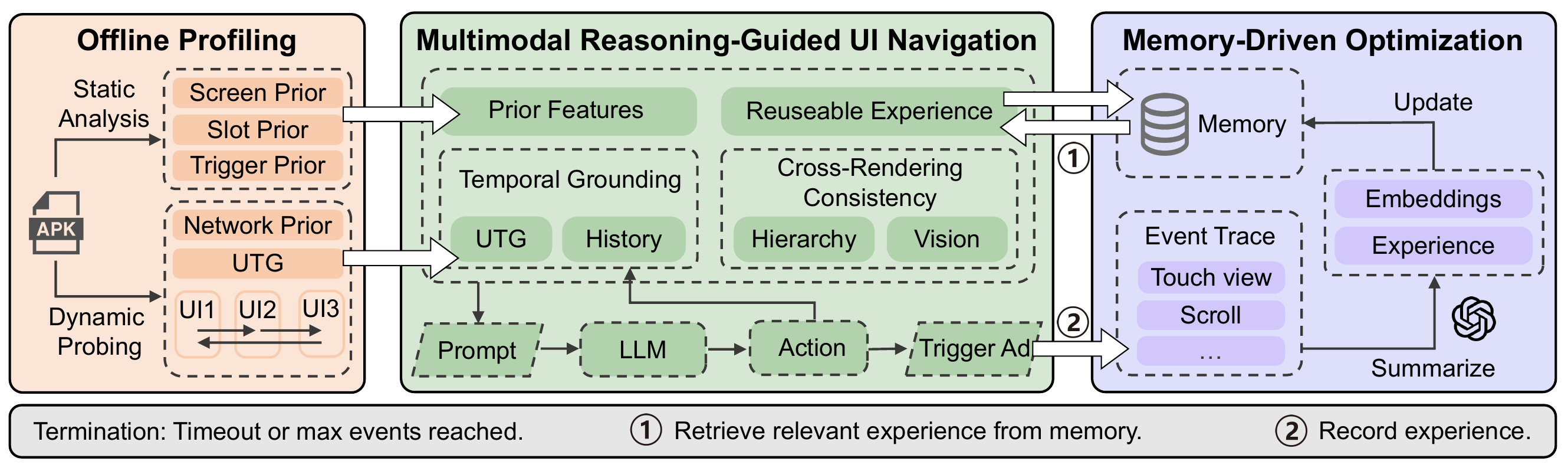}
  \vspace{-8pt}
  \caption{The system overview of \systemname.}
  \label{fig:overview}
\end{figure*}

\subsection{System Overview}

%
As illustrated in Figure~\ref{fig:overview}, \systemname operates in three tightly integrated stages: 
(1) \textbf{Offline Profiling} extracts ad\-/related priors and a coarse UTG through static analysis and lightweight exploration, providing early guidance for UI navigation; 
(2) \textbf{Multimodal Reasoning\-/Guided UI Navigation} enriches each UI state with these priors, semantic descriptions generated by the VLM, and UTG-based temporal and structural context, enabling the LLM to perform informed reasoning and avoid redundant exploration; 
(3) \textbf{Memory\-/Driven Runtime Optimization} distills successful ad\-/triggering trajectories into reusable experiences, allowing the system to recall effective strategies across similar states and thereby improve both accuraccy and efficiency.





\subsection{Offline Profiling}
At the outset of ad-oriented exploration, the detector encounters a \emph{cold start}: it lacks guidance on which screens are likely to host ads, which widgets to operate, or how long to wait before ads appear. 
This absence of prior knowledge inflates the state–action space, diminishes early coverage, and wastes exploration budget, with the additional risk of trapping the process in repetitive loops~\cite{wang2021vet}.

To address this issue, we incorporate prior knowledge to guide detection, motivated by two observations from our preliminary study.
(1) We analyze ads in 100 apps from the ADGPE dataset~\cite{ma2024careful}. As shown in Figure~\ref{fig:evidence-percentage}, declarative evidence (manifest) appears in $\sim98\%$ of cases, structural evidence (layout files with ad-specific \texttt{View} classes) in 65\% of apps, and behavioral evidence (\eg, \texttt{loadAd}/\texttt{show}) in 77\%. We also find strong temporal traces: in random exploration logs, 40.9\% of ad-related network requests occur within 3–5 seconds after the triggering interaction, indicating the value of network-based evidence.
(2) We further examine ad-type distribution (Figure~\ref{fig:adtype-percentage}) and find that custom ads dominate, while embedded and popup ads are less frequent. This directly impacts evidence availability: popup ads often expose explicit API calls, embedded ads more often leave layout traces, whereas custom ads may bypass both, making single-modality detection unreliable.

%
%
%
%
%
%

\begin{figure}[t]
\centering
\hspace{-0.1cm}
\begin{minipage}[c]{0.25\textwidth}
  \includegraphics[height=2.4cm, keepaspectratio]{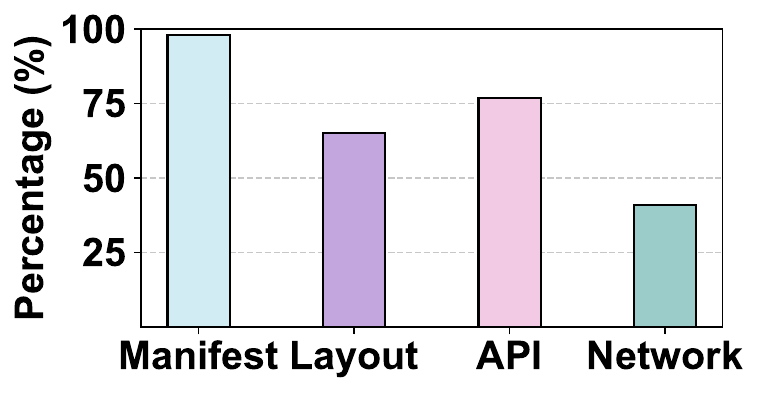}
  \vspace{-0.7cm}
  \caption{Distribution of different evidence.}
  \label{fig:evidence-percentage}
  \end{minipage}
\hspace{0.1cm}
  \begin{minipage}[c]{0.21\textwidth}
  \includegraphics[height=2.4cm, keepaspectratio]{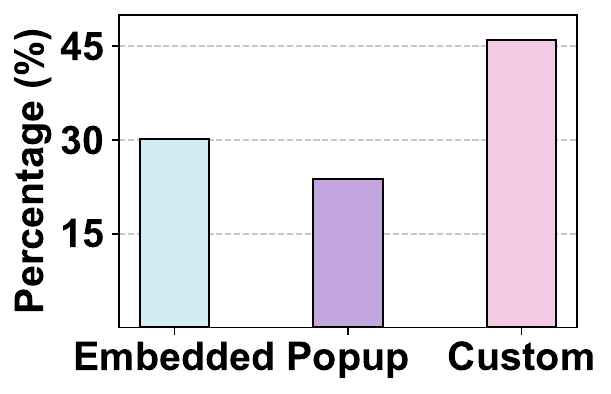}
  \vspace{-0.7cm}
  \caption{Proportions of different ad types.}
  \label{fig:adtype-percentage}
  \end{minipage}
\end{figure}

In summary, while declarative ad evidence is almost always present, structural, behavioral, and network signals vary across ad types. This motivates using multiple priors from offline profiling to narrow the search space, prioritize high-potential actions, and provide temporal guidance for ad discovery. Concretely, we perform three layers of static analysis to extract \emph{screen}, \emph{slot}, and \emph{trigger} priors, and use lightweight dynamic probing to build a coarse UTG and \emph{network} priors. We next describe how each prior is obtained.

\subsubsection{Ad-oriented Static Analysis.}
To extract actionable priors from static analysis, we follow a three-layer structure: declarative, structural, and behavioral, each uncovering complementary aspects of ad integration.

\noindent \textbf{Screen prior.}
We start from \texttt{Android\allowbreak Manifest.xml}, which encodes mandatory specifications for every app. 
Since ad SDKs must explicitly declare network permissions, register components such as \texttt{AdActivity}, and provide initialization metadata, the manifest offers reliable but coarse-grained evidence of ad integration. 
However, this evidence may be absent or misleading in adversarial cases, such as malicious apps that dynamically load ad code via reflection or obfuscate SDK entries~\cite{bhan2025dlcdroid}, or in aggregation SDKs where ad components are hidden behind generic declarations~\cite{dong2018understanding}. 
We therefore treat manifest entries only as global integration signals that contribute to \emph{screen priors}, to be cross-checked against structural and behavioral clues.

\noindent \textbf{Slot prior.}
Layout XML resources reveal the UI containers in which ads are hosted.
We parse files under \texttt{res/layout} and identify \texttt{View} classes whose package names match known ad SDK prefixes.
For each matched component, we record its resource identifier and position in the hierarchy, thereby mapping candidate ad containers to specific screens.
Beyond simple presence detection, we further infer the likely ad type (\eg, banner, interstitial, native, rewarded) from naming conventions and placement within the layout.
This process converts raw layout definitions into \emph{slot priors}, specifying where ads are likely to appear and in what form, thus providing spatial guidance for subsequent navigation and detection. 

\noindent \textbf{Trigger prior.}
Dalvik bytecode exposes operational logic beyond what manifests and layouts can show. 
We therefore propose an \emph{activity-centric attribution} model: ad-related calls (\eg, \texttt{loadAd}, \texttt{show}) and lifecycle callbacks are first detected in the code, and then traced upward through the class hierarchy until they are bound to a manifest-registered \emph{Activity}. 
This attribution ties low-level invocations to concrete screens, ensuring that behavioral evidence is localized to specific Activities rather than scattered across helper classes. 
In addition, attribution enables us to rank triggers by importance (\eg, \texttt{show*} $>$ \texttt{load*} $>$ \texttt{init*}), thereby yielding \emph{trigger priors} that indicate when ads are most likely to surface.


\subsubsection{Ad-oriented Lightweight Dynamic Probing.}
During online navigation, the UTG is gradually expanded as new states and transitions are explored, providing structural and temporal context for interaction. 
However, when an app is first analyzed, such context is largely unavailable, making the early phase inefficient. 
To mitigate this, we construct a coarse UTG using DroidBot~\cite{li2017droidbot} during offline profiling. 

\noindent \textbf{Coarse UTG construction.}
Although existing static analysis tools~\cite{azim2013targeted, lai2019goal, yang2018static} can in principle generate similar graphs, they often incur high computational costs, produce redundant paths due to over-approximation, and fail to capture runtime-specific behaviors.
We therefore adopt a lightweight dynamic random exploration that quickly samples representative transitions, yielding a more realistic and computationally efficient preliminary UTG.
\blue{The UTG provides a structural model of the state space, where nodes carry activity-level metadata and edges encode triggering events.}

\noindent \textbf{Network prior.}
Random exploration traces are noisy, so we post-process them by aligning dynamic states with static ad clues. During coarse UTG construction, \systemname captures runtime events from the Android system log (\texttt{logcat.txt}), then extracts ad-related network requests by matching known ad domains and characteristic parameters. Since each DroidBot event is timestamped, we correlate interactions with subsequent requests within a short window, linking UTG states to ad-loading activity. This yields \emph{network priors} that estimate the likelihood and timing of ad activation, and together with screen, slot, and trigger priors, provides structured evidence to guide online navigation.

\subsection{Multimodal Reasoning--Guided UI Navigation} 
As described in Section~\ref{sec:opportunity}, ad-oriented UI navigation must reason over heterogeneous signals across modalities and time scales, rather than relying on single-modality or single-step decisions.
To achieve this goal, we propose a multimodal reasoning algorithm for UI navigation.

\subsubsection{Decision Policy}
At each step $t$, \systemname instantiates an LLM-based decision policy
$\pi_{\text{MANA}}(\cdot \mid Z_t)$ via in-context reasoning.
We construct an augmented context
\begin{equation}
    Z_t = \Phi(M_t, V_t, \Sigma, H_t, \mathcal{E}_{\text{mem}}, \mathcal{G}_t),
\end{equation}
where $M_t$ is the current UI metadata, $V_t$ contains optional visual cues from the screenshot,
$\Sigma$ denotes offline priors, $H_t$ is the interaction history,
$\mathcal{E}_{\text{mem}}$ are retrieved cross-app experiences, and $\mathcal{G}_t$ is the local neighborhood of the UTG.

\noindent
\textbf{Normalized actionable observation.}
Instead of prompting the LLM with raw UI trees, \systemname converts the raw observation
$o_t=(M_t,V_t)$ into a structured actionable set $x_t$.
Each element in $x_t$ corresponds to an interactable widget with consistent attributes (\eg, type, text,
content description, bounds, and optional visual semantics), yielding a normalized action space across
different rendering pipelines.
The policy output is constrained to selecting one executable action $a_t$ from $x_t$ (plus a small set of
global actions such as \texttt{back}/\texttt{scroll}), which preserves executability and reduces hallucinated actions.

\begin{figure}[h]
\includegraphics[width=1.0\linewidth]{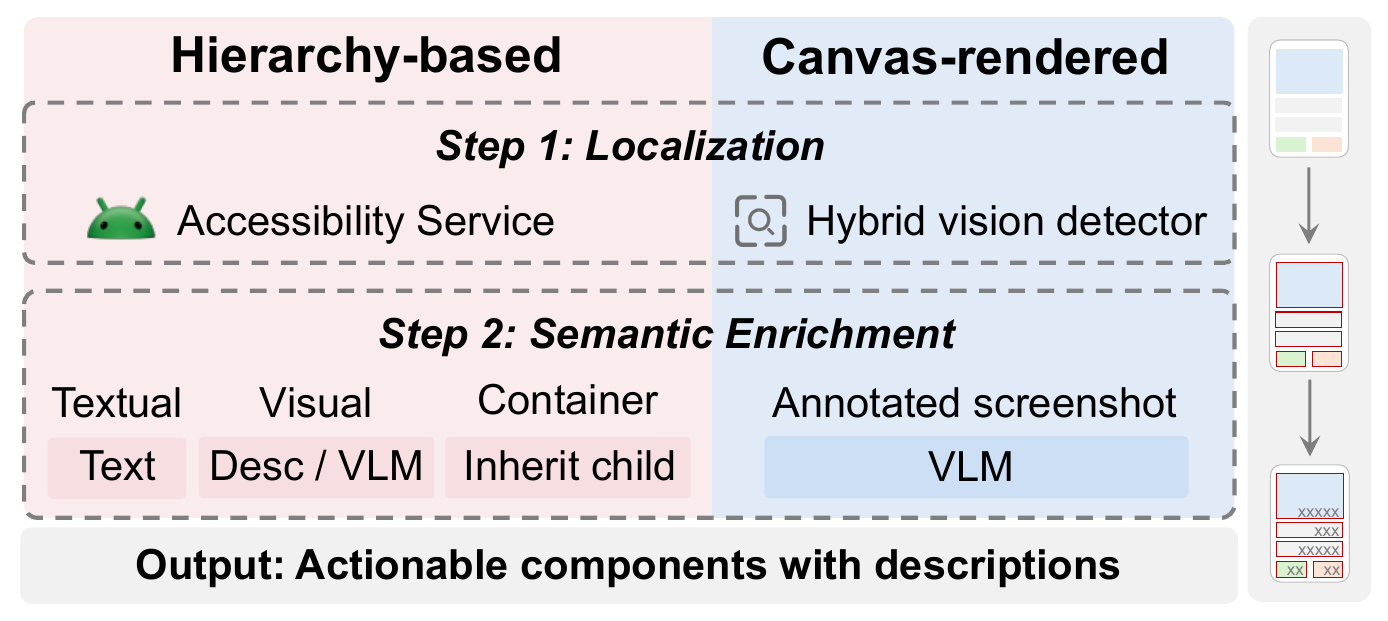}
  \vspace{-0.8cm}
  \caption{Cross-rendering consistency scheme, which combines Accessibility Service with visual techniques to produce semantically described actionable components for hierarchy-based and canvas-rendered apps.}
  \label{fig:ui}
\end{figure}

\begin{figure}[t]
\includegraphics[width=0.48\textwidth]{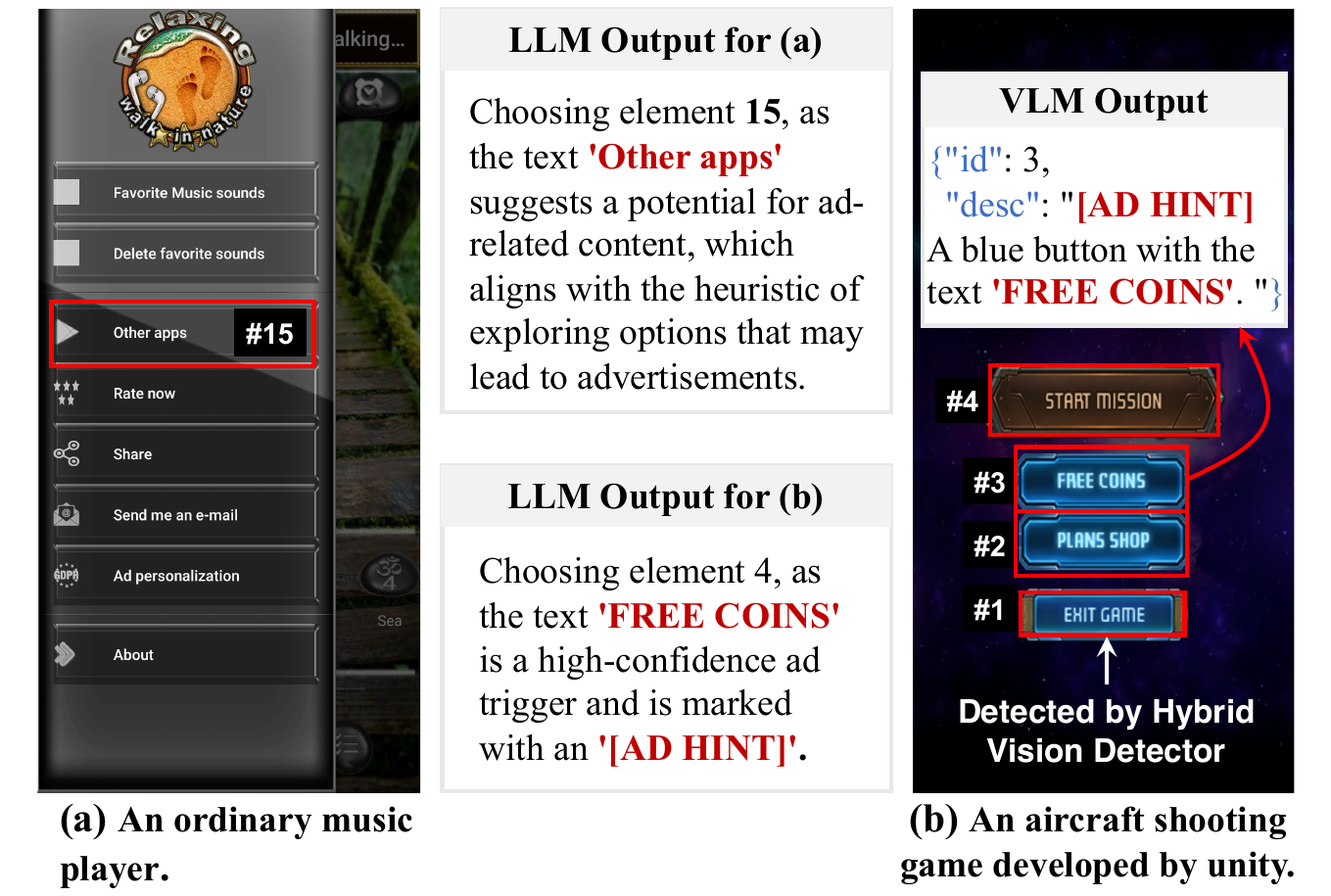}
  \vspace{-0.8cm}
  \caption{Examples of \systemname reasoning on a single UI screen. (a) a view-metadata case where hierarchical metadata identifies the ``Other apps'' button as a promotional redirection; (b) A visual cues case~(\eg, ``FREE COINS'') in a rendering app reveal advertising intent.}
  \label{fig:vlm-case}
\end{figure}

\noindent \textbf{Per-step outputs and temporal guidance.}
Given $Z_t$ and $x_t$, the LLM outputs (i) an action choice $a_t$, (ii) a brief rationale, and
(iii) an instantaneous ad-relevance score $\hat{s}_t \in [0,1]$ for Eq.~(\ref{eq:ema}).
After executing $a_t$, \systemname updates the UTG with the observed transition and smooths $\hat{s}_t$
into a node-level belief via EMA for temporal guidance and loop avoidance.
The formal definition of $Z_t$ and the coverage analysis are detailed in Appendix~\ref{app:mathematical}.

\noindent 
\textbf{Termination.}
The navigation loop terminates when an ad is detected, or when a predefined interaction budget $N_{\max}$ or time budget $T_{\max}$ is exhausted.

\noindent \textbf{Design principles.}
We realize this policy through via principles: (i) \emph{Cross-Rendering Consistency} to normalize
hierarchy-based and canvas-rendered UIs into actionable components, (ii) \emph{Temporal Grounding} to leverage
history and UTG structure, and (iii) \emph{Heterogeneous Signal Fusion} to integrate priors, runtime observations,
and reusable experiences into a unified prompt for robust decision-making.

\subsubsection{Cross-Rendering Consistency.}  
To handle rendering heterogeneity, \systemname unifies UI identification 
into two stages: \emph{localization} of candidate interactive regions and 
\emph{semantic enrichment} of their roles. 
As depicted in Figure~\ref{fig:ui}, for hierarchy-based apps, localization is obtained from the Android Accessibility Service, which exposes widget metadata such as class names, resource identifiers, and textual labels. 
With such metadata, the LLM can infer potential advertising intent~(Figure~\ref{fig:vlm-case} (a)).
However, canvas\-/rendered apps lack hierarchical metadata and thus 
require vision-based inference to recover actionable components.
To this end, \systemname introduces a \emph{hybrid vision detector} that integrates a domain\-/adapted deep detector\footnote{Instantiated with a fine-tuned YOLOv11 model in our prototype} with a heuristic visual analyzer~\cite{li2017droidbot}. 
The deep detector, tuned for canvas\-/rendered UIs, ensures high recall on domain\-/specific widgets, while the analyzer leverages lightweight geometric and pattern\-/based heuristics to augment coverage on general\-/purpose screens and low\-/complexity elements. 
The outputs of hybrid vision detector are unified through region consolidation, after which each candidate region is enriched with concise semantic captions by the VLM and subsequently provided to the LLM for downstream reasoning~(Figure~\ref{fig:vlm-case} (b)).
Our design ensures that both canvas-based and hierarchy\-/based UIs are normalized into a consistent abstraction of actionable widgets.


\noindent \textbf{Selective Vision Invocation.}
\blue{In addition to using VLMs to enrich the semantics of canvas-rendered apps, we also apply them to enhance widget level information in hierarch-based apps. However, }
\red{Since }applying VLMs to all UI elements incurs prohibitive costs in both monetary and computation, \blue{so }\systemname{} adopts a lightweight pre-filter based on UI hierarchy metadata. 
\red{Visual analysis is only triggered when handling widgets with ambiguous or missing metadata, particularly \texttt{ImageView} or canvas-based elements~(\eg, Unity) with vague identifiers, unusual layout positions, or suspicious clickability.}
\blue{Specifically, in hierarchy-based apps, VLM invocation is triggered only for widgets with ambiguous or missing metadata, particularly media container components such as \texttt{ImageView} that lack descriptive attributes. }
\red{In addition, VLMs are applied when metadata or coarse UTG construction is insufficient.}
\blue{Overall, \systemname invokes VLMs to perform semantic enrichment only on candidate regions identified by the hybrid vision detector in canvas-rendered apps or on widgets with missing metadata in hierarchy-based apps.}
This selective invocation strategy ensures robust ad detection while avoiding the overhead of applying VLMs to every UI element.
\blue{}

\subsubsection{Temporal Grounding.}
Ad-triggering behaviors are temporal: a click may only surface an ad after a short delay, repeated visits to the same screen may be required, and lack of temporal awareness can lead to redundant loops.
To capture these dynamics, we design a unified context manager that integrates recent interaction history with the evolving UTG. 
The history buffer records short-term trajectories, including executed events and ad outcomes, and uses an adaptive window to detect stagnation when the agent lingers in the same activity. 
The window expands to provide additional activity context, signaling to the LLM that it has remained in the current activity for several steps, prompting exploration of other activities and prevents excessive repetition within a single state.
The UTG \red{provides a structural model of the state space,} \blue{is} initialized with a coarse offline graph and refined incrementally online\red{, where nodes carry activity-level metadata and edges encode triggering events}.
\blue{As exploration proceeds, the UTG is continuously updated with newly observed transitions and revisit frequencies. At each decision step, it provides activities reachable within two hops from the current state, allowing \systemname to leverage structural information and decide whether to navigate toward neighbors more likely to contain ads.}
Together, these two views situate the agent in both temporal and structural context, enabling reasoning beyond single-step interactions.
This context is then consolidated into the prompt with static priors and current UI descriptions. 
The LLM both selects the next action and estimates the \red{ad likelihood} \blue{ad-relevance} of the current state\red{, producing a running estimate $\text{score}_{t+1}$ that is refined via an exponential moving average:} \blue{This estimation produces a smoothed relevance $\text{score}_{t+1}$, which is refined through an exponential moving average (EMA) mechanism:}
\begin{equation}
\label{eq:ema}
    \text{score}_{t+1} = (1-\alpha)\cdot \text{score}_t + \alpha \cdot \hat{s}_t,
\end{equation}
where \(\alpha\) is the \red{learning rate} \blue{smoothing factor} and \(\hat{s}_t\) denotes the \blue{instantaneous} ad-relevance score assigned by the LLM to the current state based on the above information. 
This dual role transforms the LLM from a reactive agent into an adaptive scorer, gradually steering exploration toward ad-relevant regions while avoiding redundant or cyclic behavior.

\subsubsection{Heterogeneous Signal Fusion.}
Even with consistent UI representations and temporal grounding, screen-only decisions remain fragile: visually similar components can lead to different outcomes, and short trajectories may miss deeper ad logic. To address this, we fuse heterogeneous signals from offline priors, runtime observations, and cross-app experiences into a unified prompt.
Concretely, we preprocess static analysis into a three-level knowledge base: (i) \emph{activity-level} signals that mark ad-related screens, potential triggers, and previously observed ad hosts; (ii) \emph{component-level} matches between visible widgets and ad-related classes/resource IDs; and (iii) \emph{global} priors such as ad libraries and network domains for consistency checks. These priors are combined with runtime UI layout and VLM semantics, as well as UTG and history context. We further retrieve memory-based experiences via embedding similarity and inject the distilled heuristics into the prompt. Together, this fusion yields a holistic decision context that improves navigation reliability, reduces redundant interactions, and increases ad discovery across diverse apps.

\begin{figure}[!t]
\includegraphics[width=0.49\textwidth]{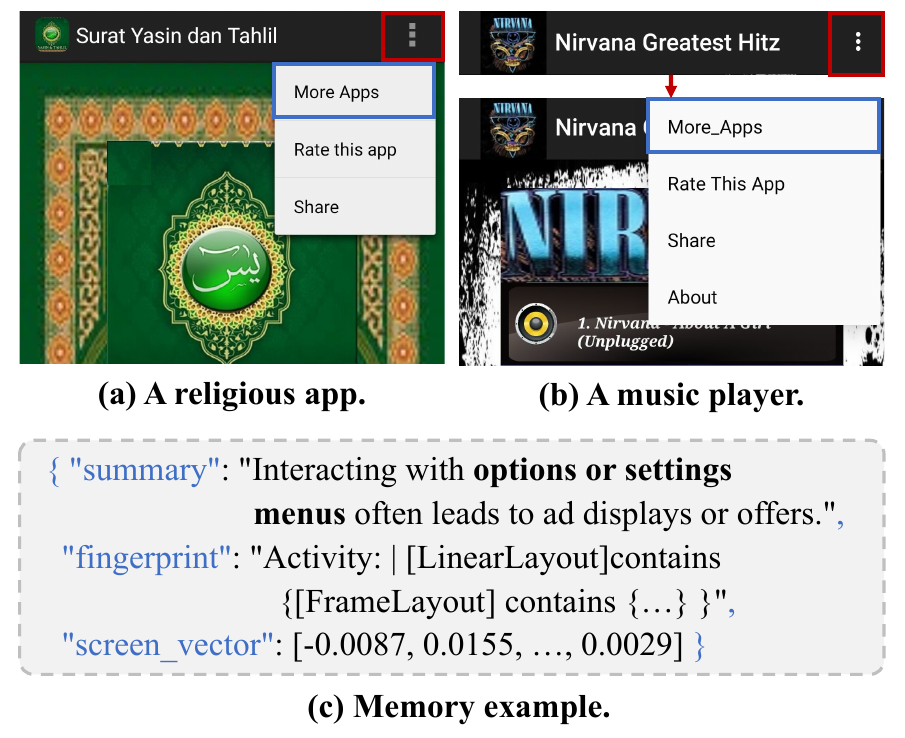}
  \vspace{-0.8cm}
  \caption{Cross-app regularities captured by memory.
(a) and (b) both exhibit recurring ad-triggering patterns via the ``Hidden Menu $\rightarrow$ More Apps'' option, (c) shows the corresponding memory entry abstracted into a reusable representation.}
  \label{fig:memory-case}
\end{figure}

\subsection{Memory-Driven Runtime Optimization}
Ad-triggering events in mobile apps are inherently sparse and context-dependent, often concealed behind ordinary\-/looking states. 
Unlike regular UI interactions, ads may surface only after specific sequences (\eg, clicking ``watch video'' followed by a short delay), or even require repeated visits to the same screen. 
Without a memory mechanism, existing agents are forced to rediscover such trajectories from scratch, leading to redundant exploration, wasted budget, and even indefinite loops that fail to reach deeper ad-hosting states.
This limitation becomes particularly acute in large applications, where the state space is vast but genuine ad opportunities remain rare.


Our insight is that successful ad triggers exhibit reusable patterns: the same types of screens, widgets, and transitions frequently precede advertisements across diverse applications, as illustrated in Figures~\ref{fig:memory-case}(a) and \ref{fig:memory-case}(b).
Motivated by this observation, we design a memory-driven optimization layer that systematically records, abstracts, and reuses prior successes to guide runtime navigation.
Whenever a trajectory successfully leads to an advertisement, the system records it as a new \emph{experience}. 
%
Each experience consists of two parts: (i) the \emph{triggering state} immediately preceding the ad, encoded as a hierarchical fingerprint of the UI that preserves class names, resource identifiers, and textual content in the view hierarchy; 
and (ii) a \emph{high-level summary} generated by the LLM that describes the effective interaction sequence in natural language, capturing actionable heuristics. 
The fingerprint is further embedded into a high-dimensional semantic space, allowing structurally and semantically similar states across apps to be aligned.
\blue{To control memory growth and maintain high retrieval efficiency, we perform offline clustering and pruning of near-duplicate or outdated trajectories. 
Formally, let $\mathcal{E} = \{\mathbf{v}_1, \mathbf{v}_2, \dots, \mathbf{v}_N\}$ denote the set of fingerprint embeddings, and define the cosine similarity between any two embeddings as 
$\text{sim}(\mathbf{v}_i, \mathbf{v}_j) = \frac{\mathbf{v}_i \cdot \mathbf{v}_j}{\|\mathbf{v}_i\| \, \|\mathbf{v}_j\|}$. 
For a similarity threshold $\tau$, any pair $(\mathbf{v}_i, \mathbf{v}_j)$ with $\text{sim}(\mathbf{v}_i, \mathbf{v}_j) \ge \tau$ is considered redundant, and only one embedding is retained. 
This pruning procedure produces a compact set $\mathcal{E}' \subseteq \mathcal{E}$ that preserves coverage while reducing redundancy.
}

At runtime, when the agent encounters a new state, it applies the same fingerprinting and embedding to query the experience repository. 
\blue{For small to moderate repository sizes, full in-memory similarity computation is sufficient; however, for scalability, we optionally employ an approximate nearest neighbor (ANN) search. This allows the system to efficiently retrieve relevant experiences even as the number of stored trajectories grows, while preserving the fidelity of similarity-based retrieval.}
Retrieved experiences provide both a similarity score and the distilled summaries as shown in Figure~\ref{fig:memory-case}(c), which are injected into the agent's reasoning process.
This enables the system to recall past ad-triggering strategies and adapt them to the current context, converting exploration from a purely reactive process into an experience-guided one. 
As a result, the agent converges more quickly toward ad-relevant states, reduces redundant navigation, and improves efficiency in diverse application environments.

\begin{figure}[!t]
\hspace*{-0.8em}
\includegraphics[width=0.5\textwidth]{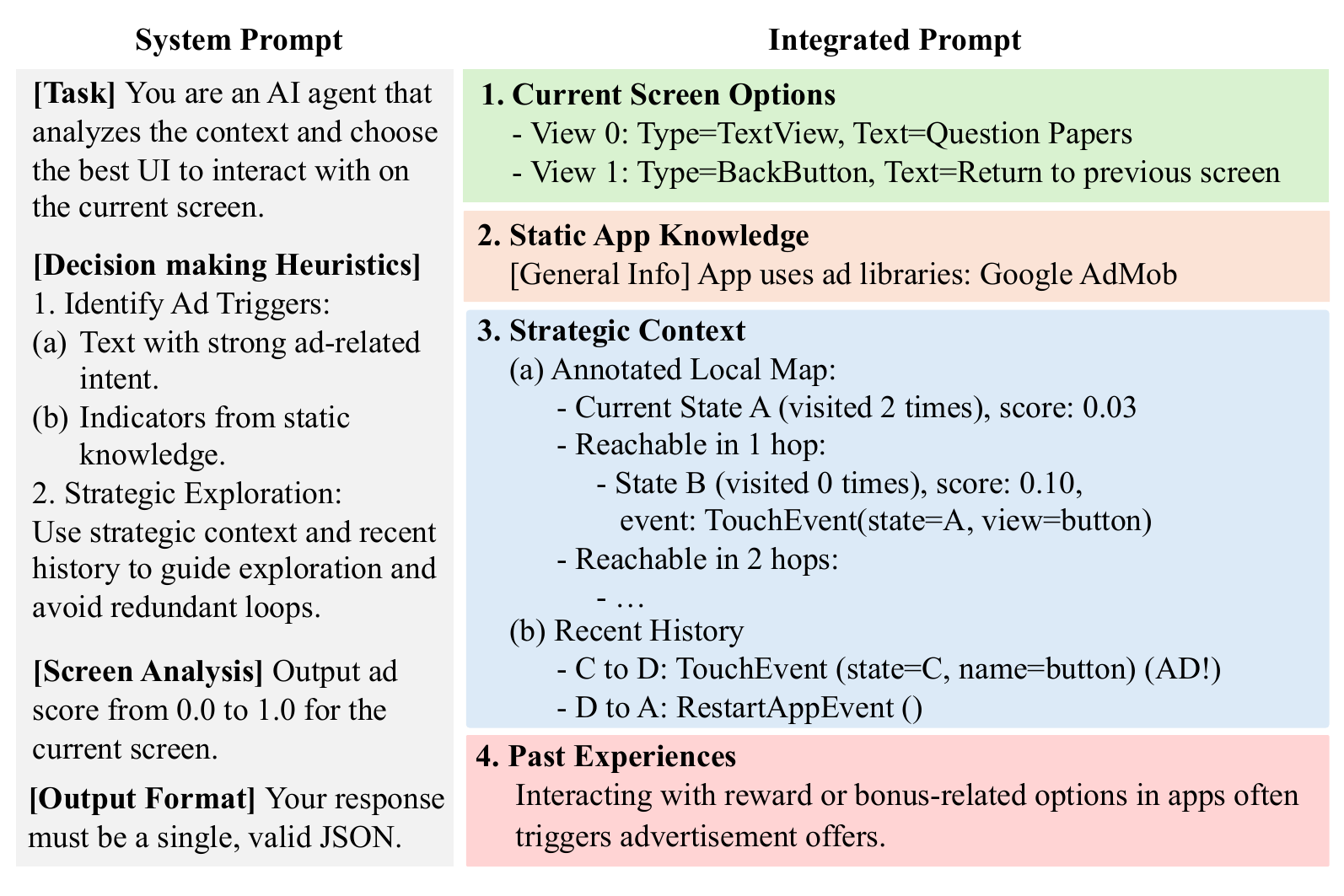}
\vspace{-0.7cm}
  \caption{An illustration of the structured prompt in \systemname. The system prompt defines the task and heuristics, while the integrated prompt supplies multi-faceted context for decision-making.}
  \label{fig:prompt}
\end{figure}


\subsection{Structured Prompt Design}
In order to inject multi-stage context into the LLM, we design a structured prompt that turns heterogeneous exploration signals into a coherent input for reasoning.
Our structured prompt has two layers. 
\textbf{System Prompt} provides stable guidance by specifying the agent’s objective, decision heuristics, and output format (\eg, identify ad triggers, avoid redundant loops, and return normalized scores). 
\textbf{Integrated Prompt} augments this backbone with contextual signals collected at runtime. 
As shown in Figure~\ref{fig:prompt}, it organizes information into four slots: (i) \textit{current screen options}, describing visible widgets and their textual attributes; (ii) \textit{static app knowledge}, distilled from offline profiling, which highlights ad-related SDKs, components, and identifiers; (iii) \textit{strategic context}, including both an annotated local UTG (reachable states, event transitions, and ad-likelihood scores) and the recent interaction history; and (iv) \textit{past experiences},
retrieved from the memory repository, which summarizes effective ad-triggering trajectories observed previously.

This multi-stage injection ensures that the LLM receives both stable objectives and structured evidence spanning current, offline, temporal, and experiential dimensions. 
Thus, the prompt enables the agent to reason jointly about where to act, what to prioritize, and how to adapt past heuristics to the current context.

%% file: chapter/implementation.tex
\section{Implementation}
\textbf{Experimental Setup. }
We deploy and evaluate our system on two commercial smartphones, an Honor V20 (Android 9) and a Huawei Mate 40 Pro (Android 12).
Each app is explored for $R=5$ independent runs with different random seeds, capped at $T_{\max}=300$ s and $N_{\max}=60$ steps per app. A 5s interval between events allows ads to fully load, which is independent of \systemnames reasoning time and can be reduced to 1-2 s in practice.
%
For fairness, all baselines are re-executed under identical budgets and device conditions. 
%
We encode hierarchical fingerprints of each screen using OpenAI's \texttt{text-embedding-ada-002}~\cite{openai_ada002} to measure structural similarity across screens.
The backbone model is GPT\-/4o mini~\cite{hurst2024gpt}, with additional comparisons to LLaMA4-Scout~\cite{meta2024llama4} and Qwen2.5-7B~\cite{qwen2p5_7b} under identical API settings. 

\red{For UI states with incomplete or ambiguous metadata, we employ a hybrid vision-based detector. Specifically, a fine-tuned YOLOv11 model (trained on the Game GUI dataset~\cite{ye2021empirical}) targets canvas-rendered and game-oriented interfaces, while a visual analyzer leveraging Android Accessibility Services reconstructs general-purpose widgets using geometric and layout heuristics. 
The two outputs are merged into a unified set of candidate regions, which are subsequently enriched with semantic captions by the VLM and forwarded to the LLM for reasoning.
}
\blue{\systemname obtains runtime information via the Android Accessibility Service.
For hierarchy-based apps, VLMs enrich semantic understanding only when widget metadata is missing.
For canvas-rendered apps, \systemname employs a hybrid vision detector: a fine-tuned YOLOv11 model (trained on the Game GUI dataset~\cite{ye2021empirical}) identifies UI widget boundaries, which are combined with results from traditional edge-based contour analysis to reconstruct fine-grained component boundaries. VLMs are then applied to semantically enrich these detected components.
Ultimately, both app types yield a unified representation of the currently interactable components, which serves as input for subsequent LLM reasoning.}
The YOLO fine-tuning step is performed once offline on a server equipped with an NVIDIA RTX~4090 GPU.
At runtime, all experiments are executed directly on mobile devices, relying only on lightweight API calls. These calls cover UI exploration, perception, and multimodal reasoning.
This design ensures that \systemname remains both practical for deployment on commodity devices and fairly comparable to baselines under identical evaluation conditions.

\noindent \textbf{Dataset. } We evaluate \systemname on two datasets: \textbf{ADGPE~\cite{ma2024careful}: }This dataset contains $100$ Android apps with $139$ ads and is relatively recent. 
However, it has several limitations: many older apps no longer function correctly; it includes clusters of highly similar apps (\eg, dictionary apps that differ only in language but share nearly identical APK structures); and it contains very few canvas-based apps, making it less suitable for evaluating our model’s handling of such interfaces.
\textbf{\datasetname: } We curate a new dataset from AndroZoo~\cite{allix2016androzoo} based on criteria ensuring both diversity and relevance: (i) apps released after $2020$ to reflect current advertising practices; (ii) confirmed presence of ad libraries, indicating a likelihood of containing ads; (iii) availability on Google Play as a proxy for app quality; and (iv) a download range between $10,000$ and $100$ million to capture apps of varying popularity. Using metadata retrieved via the AndroZoo API and subsequent filtering, we collect $100$ apps with $119$ ads, which we denote as \datasetname, including $27$ canvas-rendered apps. 

\noindent \textbf{Baselines. }
We compare \systemname with four representative dynamic exploration baselines. 

\begin{itemize}[leftmargin=*]
    \item \textit{Monkey}~\cite{monkey2023}. A system-level tool that generates random events for basic UI exploration and stress testing.
    \item \textit{DroidBot}~\cite{li2017droidbot}. A general-purpose testing tool that augments random exploration with a lightweight runtime state model for more structured navigation. In our evaluation, we employ it with a breadth-first strategy.
    \item \textit{MadDroid}~\cite{liu2020maddroid}. An ad-oriented exploration tool that prioritizes UI components likely to display ads, such as \texttt{WebView}, \texttt{ImageView}, and \texttt{ViewFlipper}.
    \item \blue{\textit{DARPA}~\cite{cai2023darpa}. A vision based method that identifies trap UI components from application screenshots using an object detection model. As DARPA does not support UI exploration, we directly collect ad containing pages and perform component detection on their screenshots.}
    \item \textit{ADGPE}~\cite{ma2024careful}. A state-of-the-art mobile ad detection framework that identifies candidate ad widgets using a keyword-driven strategy, prioritizing UI elements whose attributes (\eg, text, resource-id, class) contain ad-related terms (\eg, ``Install Now'', ``Learn More''). \blue{We compare only against its first-stage ad-navigating UI exploration, as its subsequent graph-learning module targets cross-app malware detection and is out of our scope.}
    
\end{itemize}

\noindent \textbf{Metrics.} 
We employ two complementary metrics. (1) \textit{Detection Rate} quantifies effectiveness as the number of distinct ads discovered within a fixed exploration window, with repeated refreshes of the same component counted only once. 
(2) \textit{Average Steps} quantifies efficiency as the mean number of interactions required to trigger an ad, where fewer steps indicate more direct navigation to ad components.


%% file: chapter/evaluation.tex
\section{Evaluation}

\begin{figure}[!t]
\includegraphics[width=0.48\textwidth]{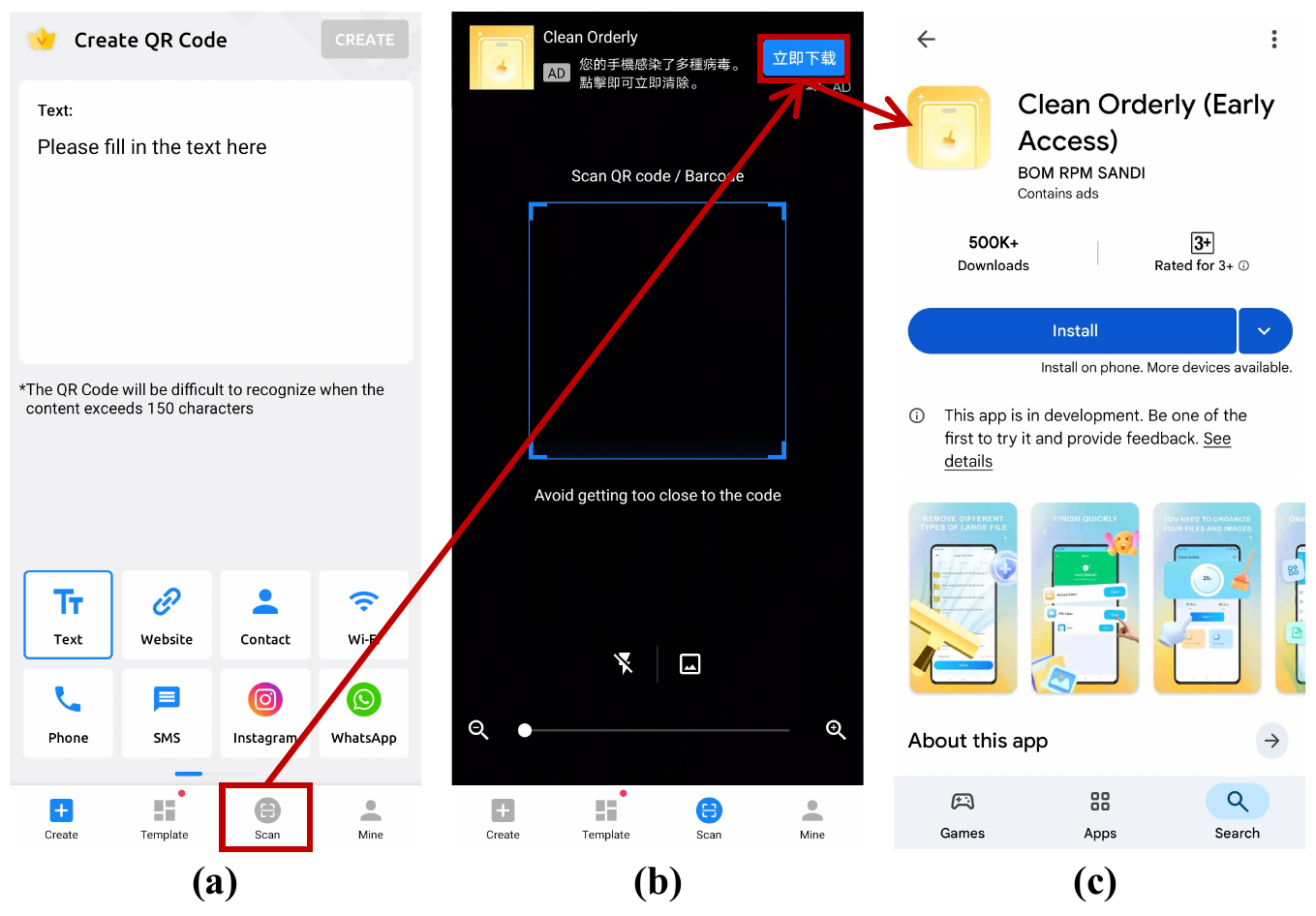}
  \vspace{-0.6cm}
  \caption{Representative success case of \systemname.}
  \label{fig:success-navigate}
\end{figure}

\begin{figure*}[t]
  \begin{minipage}{0.5\linewidth}
  \centering
  \subfigure[Detection Rate]{
  \includegraphics[width=0.486\linewidth]{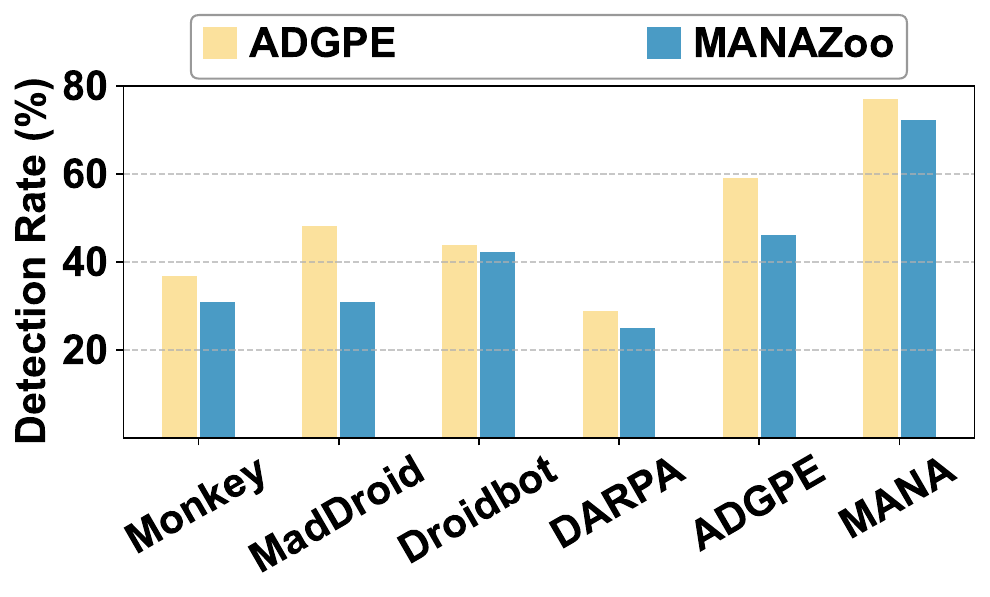} \label{fig:overall_a}}
   \hspace{-0.2cm}
  \subfigure[Average Steps]{
  \includegraphics[width=0.486\linewidth]{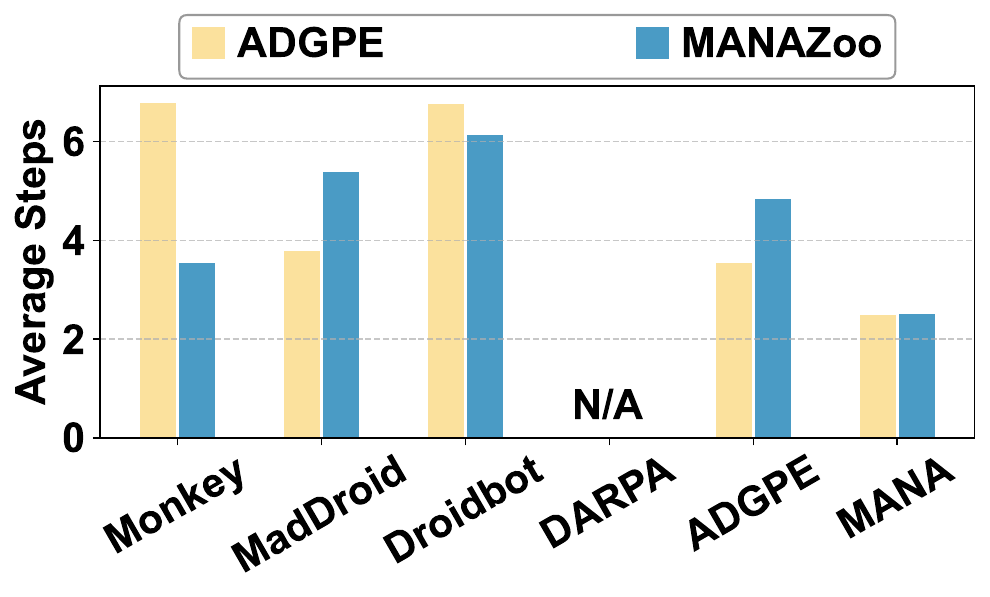}\label{fig:overall_b}}
    
    \vspace{-0.4cm}
   \caption{\blue{Overall detection performance.}}
  \end{minipage}
  \hspace{-0.2cm}
    \begin{minipage}{0.5\linewidth}
  \subfigure[ADGPE]{
  \includegraphics[width=0.486\linewidth]{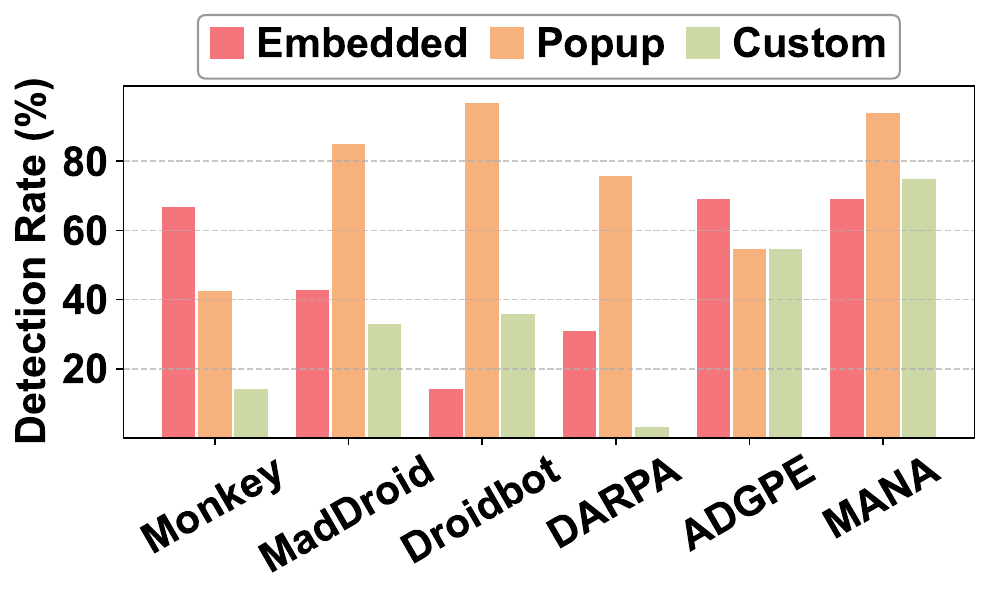} \label{fig:overall_c}}
     \hspace{-0.22cm}
  \subfigure[MANAZoo]{
  \includegraphics[width=0.486\linewidth]{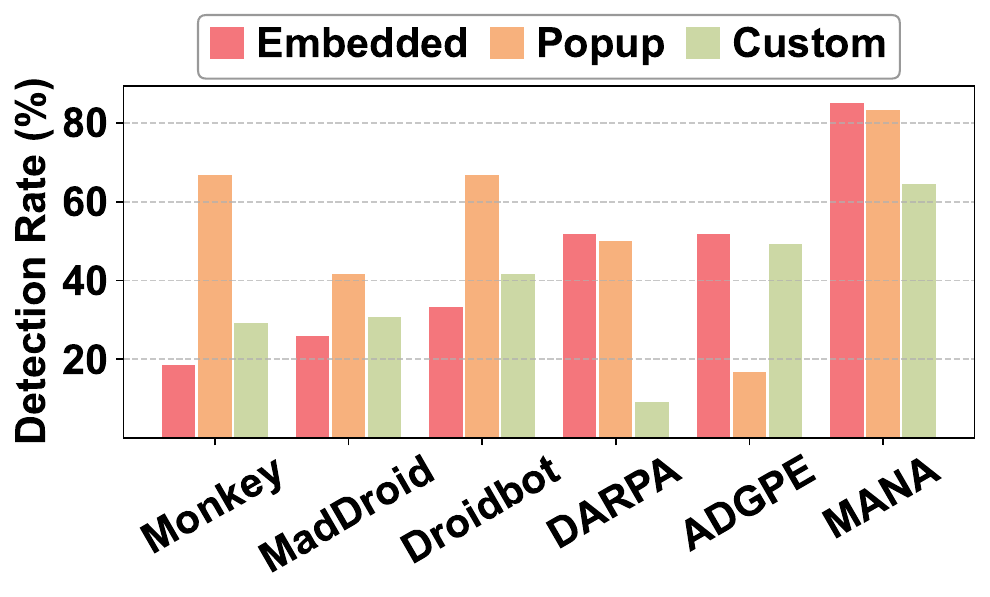}
      \label{fig:overall_d}}
       \vspace{-0.4cm}
   \caption{\blue{Category-wise detection performance.}}
 \end{minipage}

 \label{fig:baseline}
\end{figure*}

\begin{figure*}[t]
\vspace{0.1cm}
  \begin{minipage}{0.5\linewidth}
  \centering
  \subfigure[ADGPE]{
  \includegraphics[width=0.486\linewidth]{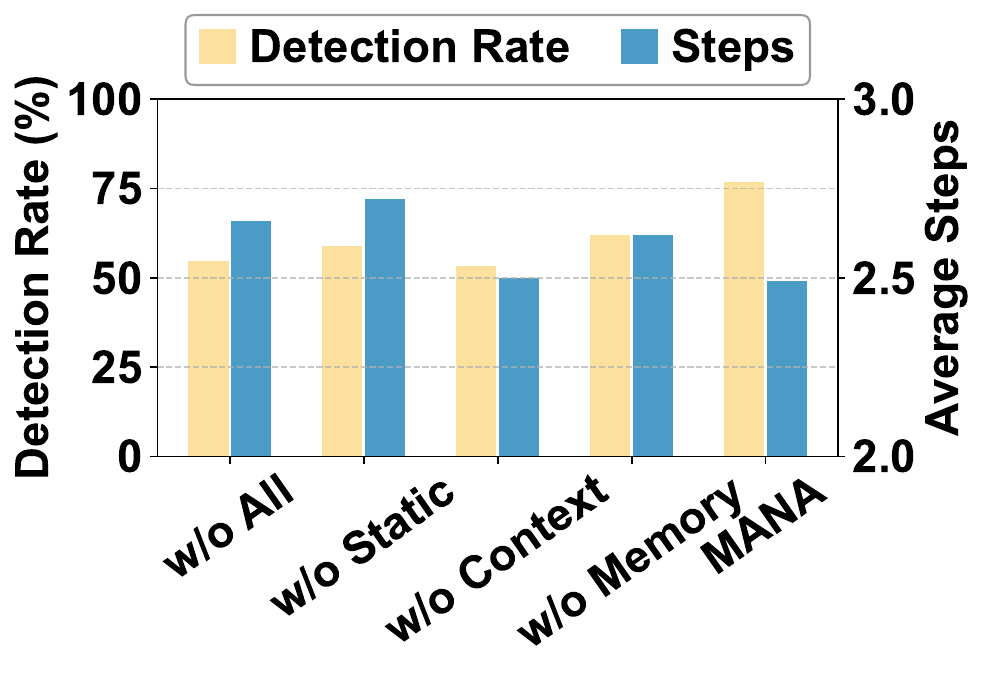} \label{fig:ablation-total-adgpe}}
   \hspace{-0.2cm}
  \subfigure[MANAZoo]{
  \includegraphics[width=0.486\linewidth]{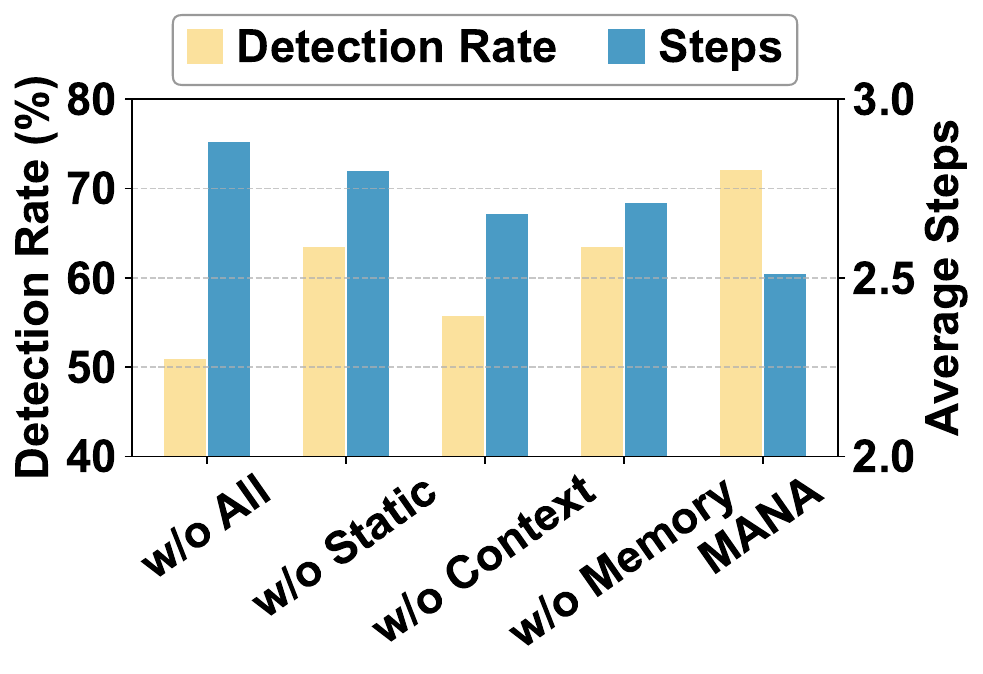}\label{fig:ablation-detail-adgpe}}
    \vspace{-0.4cm}
   \caption{Overall ablation results.}
     \label{fig:ablation1}
  \end{minipage}
  \hspace{-0.2cm}
    \begin{minipage}{0.5\linewidth}
  \subfigure[ADGPE]{
  \includegraphics[width=0.486\linewidth]{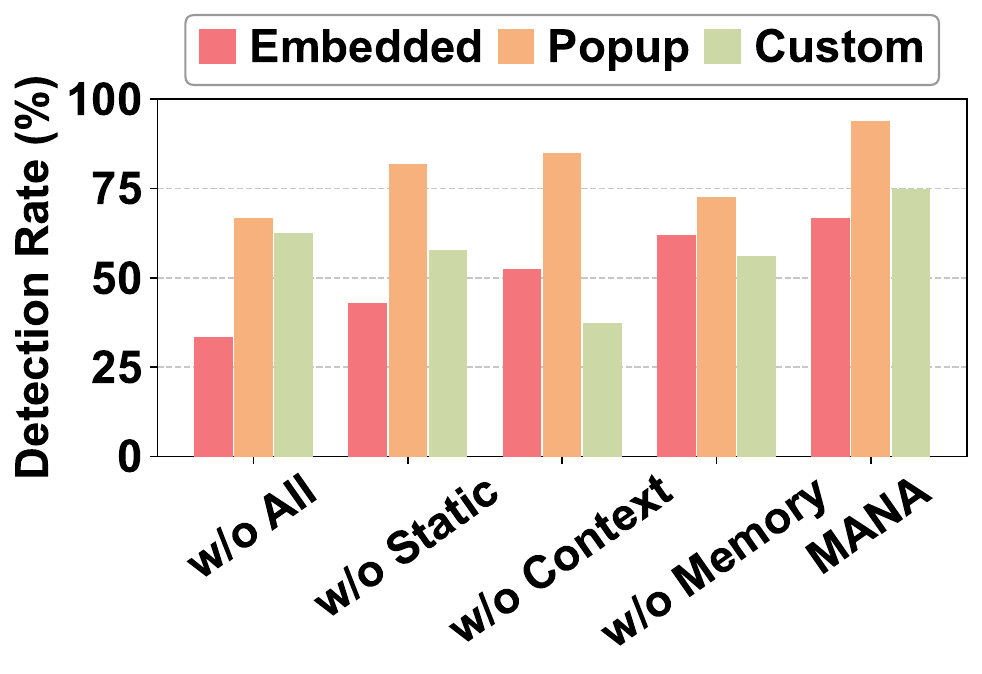} \label{fig:ablation-total}}
     \hspace{-0.22cm}
  \subfigure[MANAZoo]{
  \includegraphics[width=0.486\linewidth]{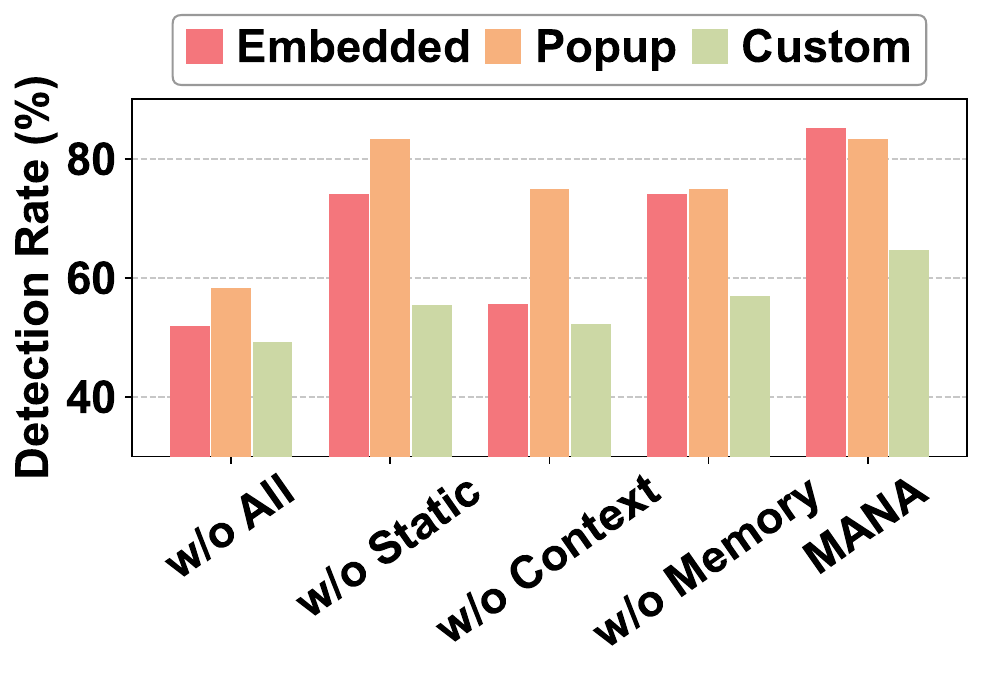}
      \label{fig:ablation-detail}}
       \vspace{-0.4cm}
   \caption{Ablation impact across ad categories.}
      \label{fig:ablation2}
 \end{minipage}

\end{figure*}

\subsection{Overall Performance}
\label{sec:rq1}
We evaluate \systemname against the baselines on both the ADGPE and MANAZoo datasets.
Figure~\ref{fig:success-navigate} shows a representative case in MANAZoo where \systemname uncovers a highly concealed ad. Instead of an obvious banner or popup, the ad is hidden behind the ``Scan'' entry of a QR code tool (a), leading through a camera interface (b) before redirecting to an external app store page (c). Such multi-step obfuscation easily defeats heuristic exploration, but \systemname efficiently pinpoints the trigger by combining heterogeneous signals, reasoning both ``\textit{where to go}'' and ``\textit{how to go}''.

In addition, Figure~\ref{fig:overall_a} shows that \systemname consistently outperforms the baselines in coverage, achieving detection rates of $77.0\%$ on ADGPE and $72.1\%$ on MANAZoo, which represent $30.5\%$–$56.3\%$ relative improvements and set a new state-of-the-art.
These gains can be better understood through a breakdown by ad category, as defined in Sec.~\ref{sec:adbackground} (Figure~\ref{fig:overall_c}–\ref{fig:overall_d}), which highlights the sources of improvement.
Popup ads are typically easier to detect because of their intrusive interaction patterns\blue{, which allows even simple baseline methods to achieve relatively strong detection performance}; however, in MANAZoo many are embedded within functional buttons, which makes them substantially harder to capture.
Notably, ADGPE underperforms in this category because its keyword-driven heuristics fail to capture popups lacking explicit textual cues, whereas \systemname attains higher detection rates by leveraging complementary multimodal reasoning.

For Embedded and Custom ads, ADGPE outperforms \red{DroidBot and MadDroid} \blue{other baseline methods}, yet still falls substantially short of \systemname.
We further observe that while baseline detection rates are often limited to $20\%$–$40\%$, \systemname consistently attains $70\%$–$80\%$, demonstrating robustness to obfuscation and non\-/standard implementations.
This contrast is particularly evident when compared with heuristic exploration methods such as DroidBot and MadDroid.
DroidBot explores apps using UI\-/state modeling with rule-based event generation, whereas MadDroid augments this approach with ad-specific media UI container heuristics.
These strategies are adequate for obvious cases but break down when ads are obfuscated or rendered through custom components.
\blue{DARPA represents a vision based approach that detects specific UI components from screenshots. While it performs relatively well on Popup ads, it exhibits clear limitations on other ad categories, particularly Custom ads, as such carefully crafted ads often lack explicit visual cues.}
In contrast, \systemname adopts an agentic multimodal reasoning paradigm that infers advertising intent from heterogeneous signals, thereby maintaining high detection performance even under obfuscation and evolving implementations.

\noindent \textbf{System Efficiency.} 
Figure~\ref{fig:overall_b} demonstrates that \systemname markedly reduces redundant UI exploration, requiring substantially fewer steps for ad detection. 
\blue{Average steps are not applicable to DARPA, as it lacks navigation capability and analyzes only the current screenshot.}
Compared with random baselines (Monkey and DroidBot), it shortens trajectories by $29.8\%$–$63.3\%$, and against heuristic baselines (MadDroid and ADGPE) by $29.7\%$–$53.8\%$, while simultaneously achieving higher detection rates. 
The average step count converges to $2.4$, approaching the practical lower bound. 
\blue{The average LLM reasoning latency is $\sim2.96$s per step, resulting in fewer interactions under a fixed time budget, yet higher detection rates, indicating more effective per-step decisions.}

%
These results indicate that efficiency gains arise from the agent's ability to infer ``where to go'', yielding a improvement of higher detection accuracy and lower exploration costs.
From a deployment perspective, fewer steps also translate into reduced computational overhead, energy consumption, and latency, underscoring the practicality of \systemname.

\subsection{Ablation Study}
\label{sec:rq2}
We further evaluate the role of heterogeneous signals by ablating static profiling, contextual reasoning, and memory (Figure~\ref{fig:ablation1}).
Without these signals, performance degrades substantially, with detection rates dropping to $\sim50\%$ and average steps rising to nearly $3$, underscoring their importance in avoiding blind exploration.
Removing static profiling markedly increases steps, \red{confirming its value in pruning redundant paths.}
\blue{reflecting that in \systemname, static analysis primarily serves as a cold start prior that accelerates early exploration and prunes redundant paths rather than acting as a strict dependency, since the system remains effective through multimodal inference even when static cues are unavailable or obfuscated.}
In contrast, eliminating contextual signal leads to detection rates dropping below $60\%$, highlighting the necessity of this heterogeneous cue for inferring implicit ad semantics.
Memory improves both detection and efficiency by retaining cross-app knowledge of effective and ineffective trajectories, thereby enhancing robustness under obfuscation.
%
%
In its full configuration, \systemname achieves the highest detection rate ($>75\%$) and the fewest steps, approaching the practical lower bound.

Figure~\ref{fig:ablation2} reports category-wise results and shows that heterogeneous signals contribute differently across ad types. Popup ads are relatively easy, but static signals still improve reliability by reducing spurious interactions. Embedded and Custom ads degrade sharply without contextual signals, indicating their importance for reasoning under weak or obfuscated cues. Memory is most helpful for Custom ads, where non-standard implementations benefit from accumulated experience. With all signals enabled, \systemname achieves the best accuracy across categories ($\approx85\%$ Popup, $\approx80\%$ Embedded, $\approx75\%$ Custom), confirming complementary roles of static profiling, contextual reasoning, and memory for accuracy–efficiency trade-offs.

%
%
%
%
%
%

\begin{figure*}[!h]
  \begin{minipage}[t]{0.3\linewidth}
    \centering
    \includegraphics[height=2.8cm, keepaspectratio]{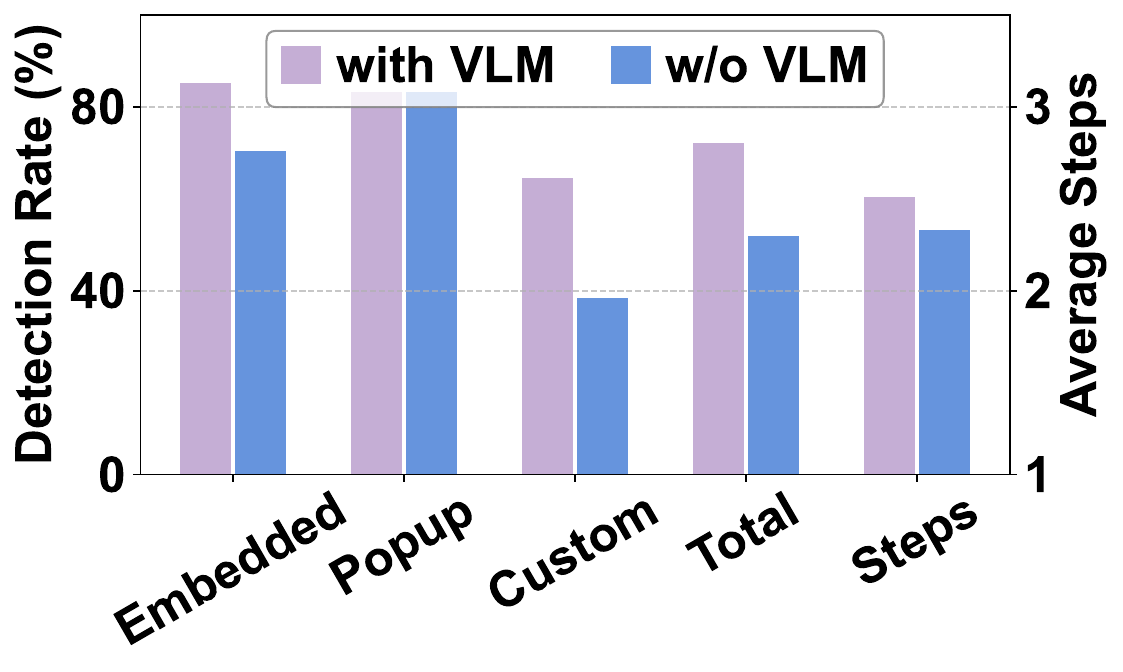}
         \vspace{-0.35cm}
    \caption{Impact of visual cues.}
    \label{fig:ablation-vlm}
  \end{minipage}
  \hspace{-0.2cm}
  \begin{minipage}[t]{0.20\linewidth}
    \centering
    \includegraphics[height=2.8cm, keepaspectratio]{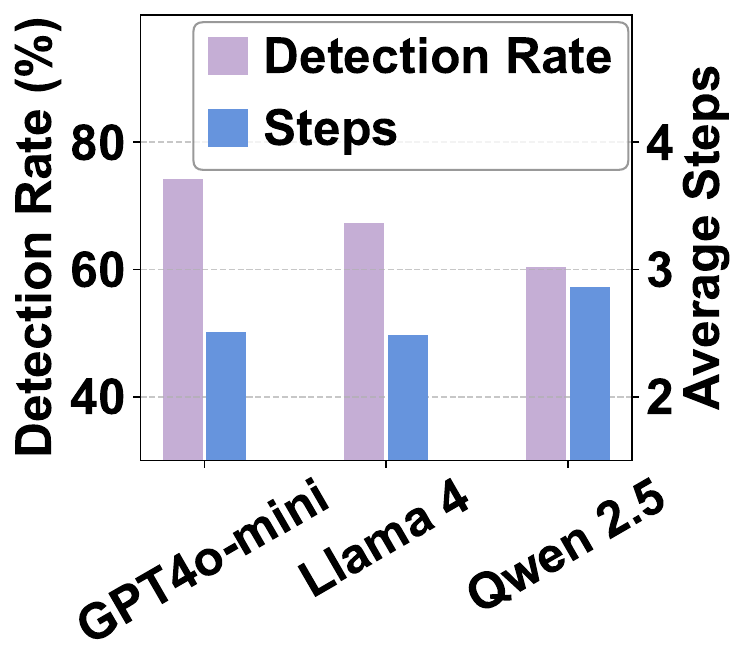}
         \vspace{-0.35cm}
    \caption{Models.}
    \label{fig:llm-comparison}
  \end{minipage}
  \hspace{-0.2cm}
  \begin{minipage}[t]{0.20\linewidth}
    \centering
    \includegraphics[height=2.8cm, keepaspectratio]{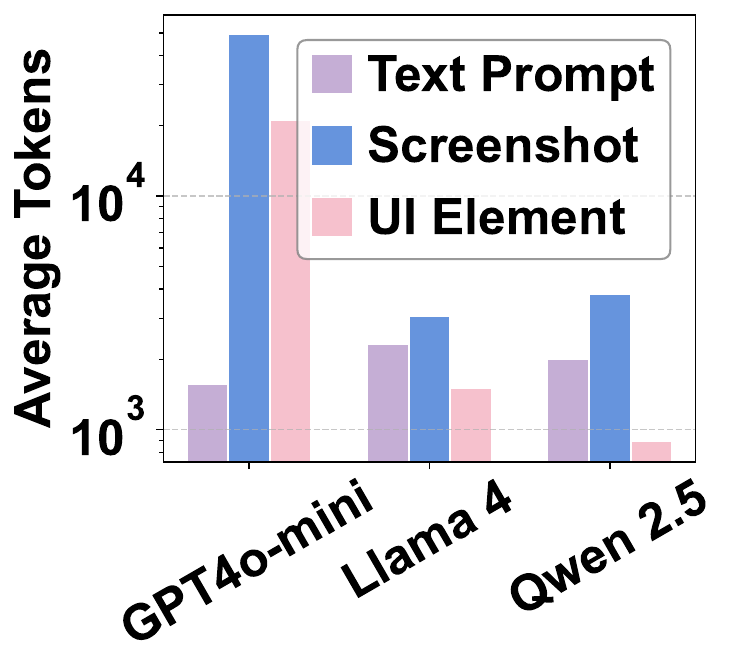}
         \vspace{-0.35cm}
    \caption{Tokens cost.}
    \label{fig:tokens}
  \end{minipage}
  \hspace{-0.2cm}
  \begin{minipage}[t]{0.3\linewidth}
    \centering
    \includegraphics[height=2.8cm, keepaspectratio]{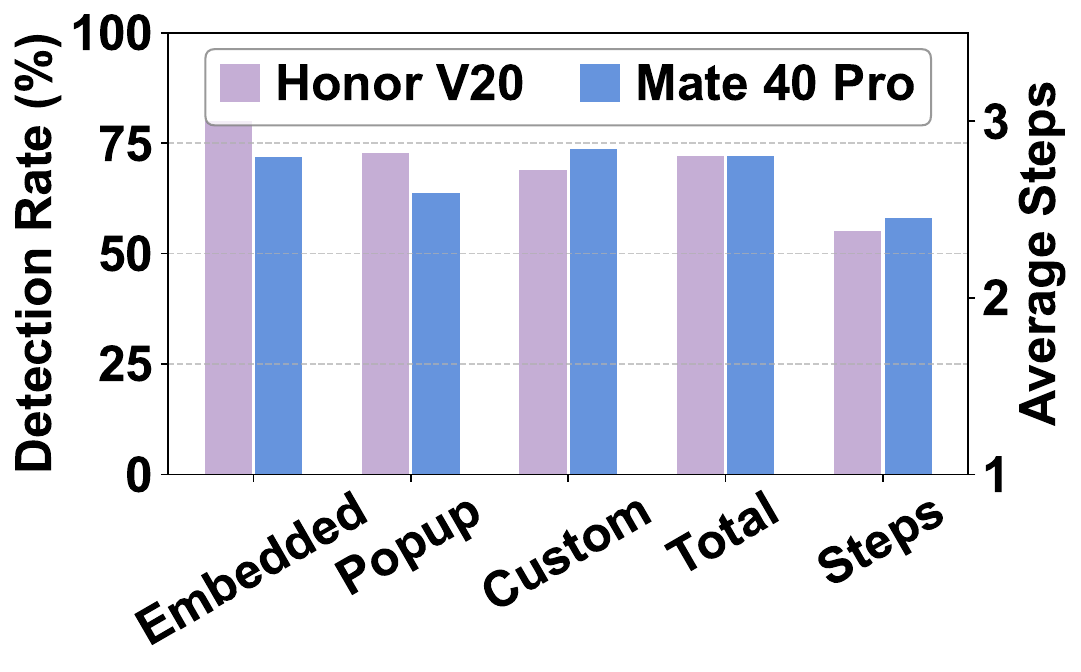}
         \vspace{-0.35cm}
    \caption{Phone comparison.}
    \label{fig:phone-comparison}
  \end{minipage}

  
  
\end{figure*}

\noindent \textbf{Impact of Visual Cues.}
As shown in Figure~\ref{fig:ablation-vlm}, adding the VLM significantly improves detection, especially for Embedded and Custom ads where explicit textual cues are limited and keyword/structure-based exploration falls short. By mapping subtle visual elements (e.g., icons, banners) to semantic meaning, the VLM boosts performance by $21.1\%$–$67.2\%$ in these categories, while gains for Popup ads are marginal given their already explicit interaction patterns.
%
%
In contrast, the impact on efficiency is less pronounced: the average steps \red{are only modestly reduced, indicating that} \blue{in the VLM-enabled setting is slightly increased. This outcome is due to the VLM's capacity to identify and explore extra interaction paths associated with visually subtle ads, which are otherwise overlooked. Therefore, the} VLM mainly expands semantic coverage rather than \red{drastically} shortening navigation paths.
Overall, these results confirm that visual clues are indispensable for capturing heterogeneous ad semantics, though its contribution lies primarily in robustness of detection rather than step-wise efficiency.

\subsection{Comparison across Reasoning Models}
We evaluate \systemname with three base models: GPT-4o-mini (8B), Llama-4 (17B), and Qwen-2.5 (7B). Figure~\ref{fig:llm-comparison} shows performance tracks reasoning ability more than model size. GPT-4o-mini performs best (74.3\% detection, 2.5 steps), Llama-4 is moderate (67.3\%, 2.5 steps), and Qwen-2.5 is lowest (60.4\%, 2.9 steps). Notably, GPT-4o-mini and Qwen-2.5 are similar in size, yet GPT-4o-mini is more accurate and efficient, reinforcing that cross-modal reasoning, not parameter count, is key for robust ad detection.

%
%
%
%

Figure~\ref{fig:tokens} decomposes token consumption into text-prompt, UI-element, and screenshot tokens. 
\blue{UI element tokens represent widget screenshots without metadata in hierarchy-based apps, whereas screenshot tokens represent full screen captures in canvas-rendered apps.}
Screenshot tokens dominate the overall budget, often one to two orders of magnitude larger than the other components, making them the primary cost driver. 
While text and UI tokens remain relatively stable across models, the variance in screenshot tokens aligns closely with performance, reflecting the importance of visual grounding. 
Our design reduces unnecessary image invocations by leveraging static profiling, metadata, and memory first, invoking screenshots only when necessary. 
This adaptive routing lowers the cost of vision-unrelated queries by 56.8\%–68.4\% in token usage compared with a VLM-first strategy, while maintaining or even improving coverage, thereby offering a more cost-efficient solution.

\blue{Furthermore, \systemnames performance is influenced by both the LLM and VLM. Qwen-2.5 and Llama-4 both use relatively low default image processing resolution, resulting in fewer screenshot and UI element tokens, as shown in Figure~\ref{fig:tokens}. However, Llama-4, being a larger 17B model with stronger reasoning ability, achieves better detection performance than Qwen-2.5, as reflected in Figure~\ref{fig:llm-comparison}. This illustrates a practical trade-off between default processing resolution, token consumption, and reasoning capability.}

\subsection{Comparison across Smartphones}
Figure~\ref{fig:phone-comparison} compares Honor V20 and Mate 40 Pro. \systemname maintains consistently high detection on both, with only minor category-wise differences. Mate 40 Pro is slightly more accurate on Embedded and Custom ads ($\approx2$–$3\%$), likely due to more stable rendering and screenshot processing, while Popup accuracy is essentially identical across devices. Average exploration steps are also similar ($\approx2.5$), suggesting \systemname is largely hardware-agnostic and remains reliable even on resource-constrained phones.
\begin{figure}[!t]
\includegraphics[width=0.48\textwidth]{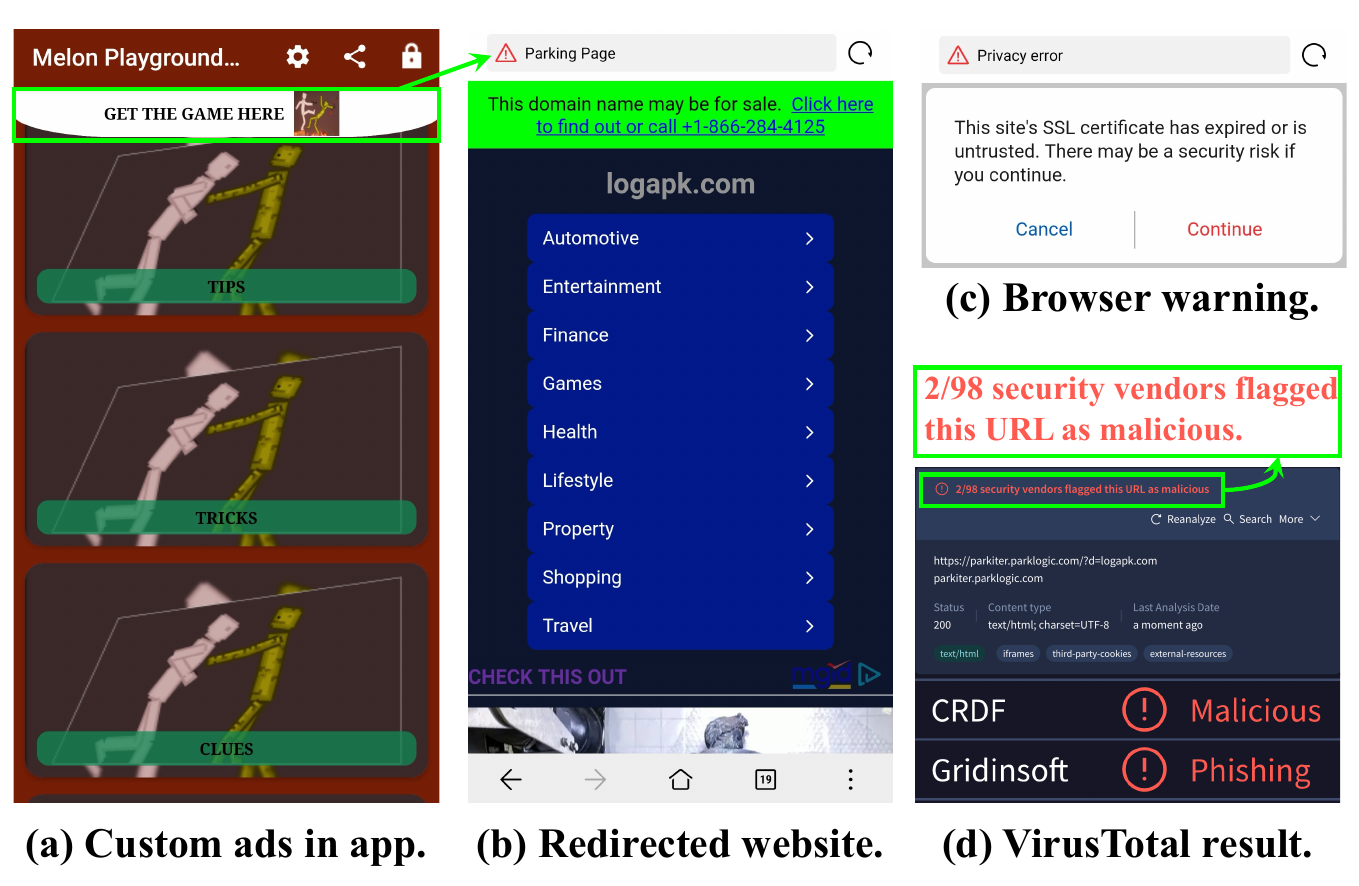}
  \vspace{-0.8cm}
  \caption{Automatic detection and analysis of malicious advertising by \systemname. (a) A custom in-app ad is detected. (b) The interaction triggers a redirection to a suspicious website. (c) The system captures the ensuing browser warning, and (d) submits the URL to VirusTotal, which confirms it as malicious.}
  \label{fig:malicious-ad}
  \vspace{-0.1cm}
\end{figure}
\begin{figure}[!t]
\includegraphics[width=0.48\textwidth]{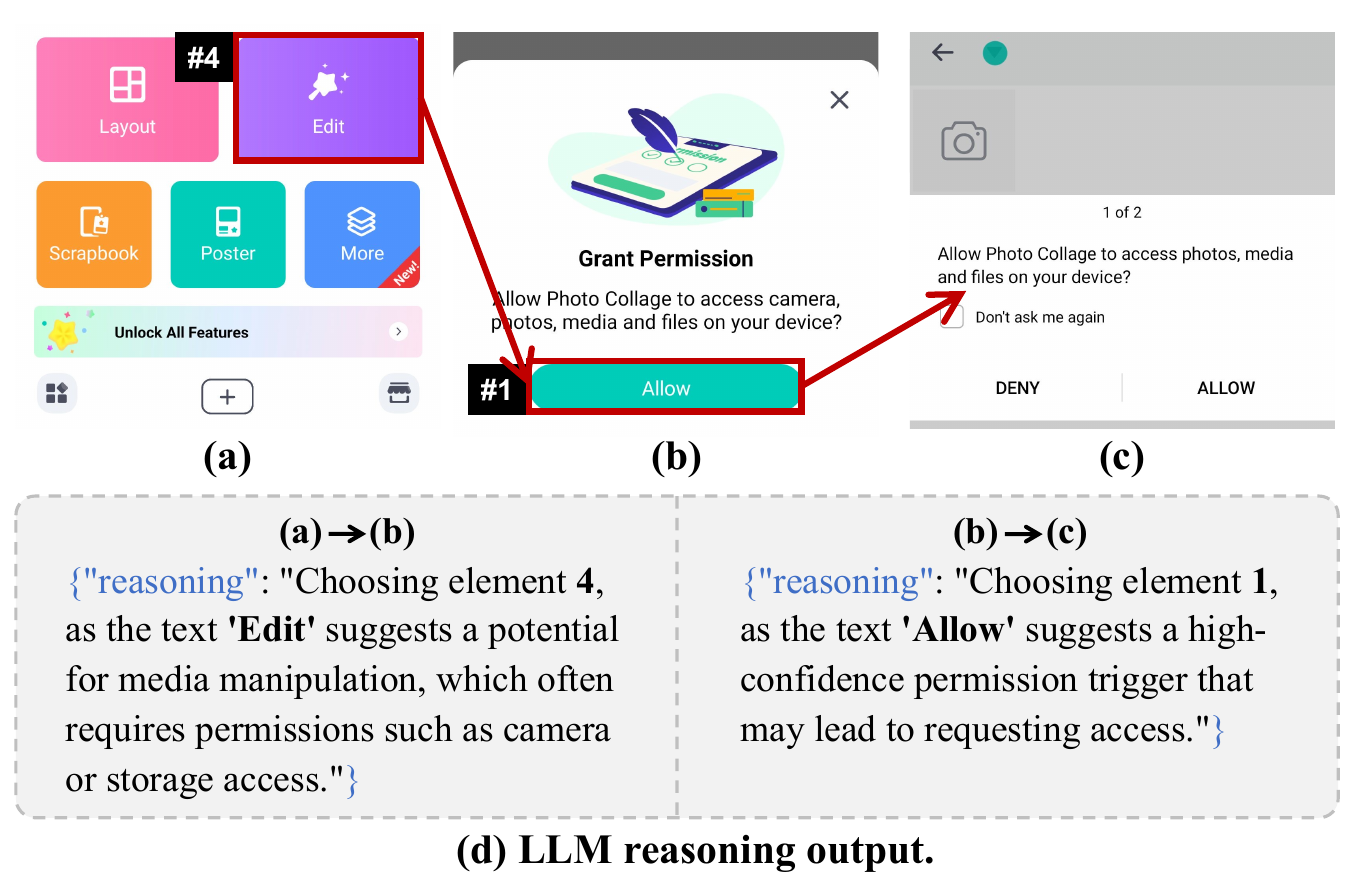}
  \vspace{-0.8cm}
  \caption{Case study on permission abuse.}
  \label{fig:permission}
  \vspace{-0.2cm}
\end{figure}

\subsection{Case Studies}
We present a set of case studies to illustrate the practical behaviors captured by \systemname. 
\label{sec:rq4}
\subsubsection{Malicious ad regulation. }
This case shows \systemnames ability to go beyond conventional ad detection and support malicious advertising oversight. Starting from an in-app custom ad (Figure~\ref{fig:malicious-ad} (a)), \systemname traces the user-triggered redirection to an external website (Figure~\ref{fig:malicious-ad} (b)), automatically captures subsequent browser warnings (Figure~\ref{fig:malicious-ad} (c)), and verifies the target URL using security services such as VirusTotal (Figure~\ref{fig:malicious-ad} (d)). This end-to-end capability highlights two key regulatory values: first, the ability to uncover deceptive ad flows that cross the app–web boundary and escape static analysis; second, the capacity to automatically link ad-triggering interactions to concrete evidence of security risks, such as phishing or malware distribution. By bridging detection with forensic validation, \systemname provides a practical foundation for auditing compliance with platform policies and industry standards, thereby extending its utility from ad transparency to mobile ecosystem security.

\subsubsection{Extending where-oriented reasoning to permission abuse.}
Figure~\ref{fig:permission} presents a case of UI-to-permission escalation. The app first displays a benign interface (Figure~\ref{fig:permission}(a)), where the Edit button seems harmless. Once selected, however, the workflow quickly escalates: a permission grant dialog (Figure~\ref{fig:permission}(b)) is followed by a system-level request for access to local photos and media (Figure~\ref{fig:permission}(c)).
\systemnames reasoning (Figure~\ref{fig:permission}(d)) highlights the analysis process. It interprets Edit not as a random click, but as a semantic cue likely linked to sensitive resources (\eg, storage or camera). This inference guides exploration ``toward'' the path where escalation may occur. When the Allow option appears, the system correctly recognizes it as a privilege-escalation trigger, thereby reconstructing the hidden chain from benign entry to sensitive permission abuse.
This case demonstrates that \systemnames where-oriented reasoning generalizes beyond advertising: by inferring UI paths likely to expose sensitive operations, it can also uncover evasive behaviors such as permission misuse.

\begin{figure}[!t]
\includegraphics[width=0.5\textwidth]{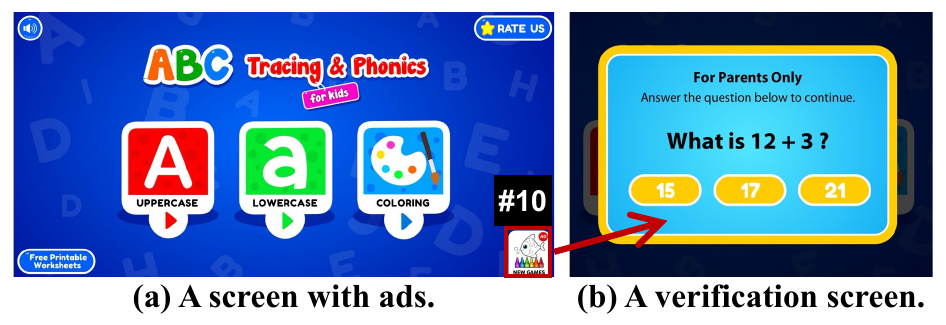}
  \vspace{-0.9cm}
  \caption{Representative failure case of \systemname.}
  \label{fig:failure-case}
\end{figure}

\subsubsection{When multimodal reasoning fails? }
As shown in Figure~\ref{fig:failure-case}, after clicking an ad button, the app jumps to a parental verification screen that requires solving a math puzzle. Unlike typical ad-related interfaces, this screen provides no visual or textual ad cues. The VLM extracts literal elements (e.g., ``15'' and ``21''), while the LLM often fails to infer that the puzzle must be solved to proceed. Thus, the agent tends to backtrack instead of progressing.
This is not a corner case: ad triggers can be placed behind puzzles, age checks, or multi-step verification flows. Handling such cases requires hybrid reasoning that infers the hidden task requirement and tracks the goal across steps. More broadly, this points to future work: moving from locating ``where ads might be'' to identifying ``what reasoning steps are needed'' to reveal them.

%% file: chapter/related_work.tex
\section{Related Work}
\noindent \textbf{Detecting Mobile Ads on Android. }
Research on mobile ad detection has evolved from static analysis, which inspects manifests, layouts, or bytecode to flag advertising SDKs~\cite{lee2019adlib, feldman2014manilyzer, lee2016hybridroid, xinyu2023andetect, yang2015static, chen2019storydroid, liu2022promal, xiao2019iconintent}, to dynamic exploration that uncovers ads at runtime through automated UI navigation. Static systems such as Gator~\cite{yang2015static}, Frontmatter~\cite{kuznetsov2021frontmatter}, and IconIntent~\cite{xiao2019iconintent} link callbacks to UI widgets, while later work (StoryDroid~\cite{chen2019storydroid}, Promal~\cite{liu2022promal}) builds transition graphs to model inter-/component interactions. Dynamic tools range from general-/purpose explorers (\eg, Monkey~\cite{monkey2023}, DroidBot~\cite{li2017droidbot}) to ad-specific strategies, including breadth-first exploration with HTTP hooking~\cite{liu2020maddroid}, network-augmented analysis~\cite{dong2018frauddroid}, keyword/metadata prioritization~\cite{ma2024careful}, and computer-/vision–based detection~\cite{cai2023darpa}. Overall, these approaches often rely on single-source cues and coarse navigation, making sparse, concealed, or visually ambiguous ads hard to uncover. 

\noindent \textbf{LLM/VLM-based Mobile GUI Agent. }
Recently, numerous LLM/VLM-based mobile agents have been proposed~\cite{liu2024make, zhao2025llm, wen2024autodroid,gou2025navigating,ariaui,qin2025ui,wan2024omniparser,you2024ferret}. Some focus on prompt-centric design, improving task execution via structured prompting and memory (\eg, AutoDroid~\cite{wen2024autodroid}, MobileGPT~\cite{lee2024mobilegpt}). Others adopt perception-enhanced strategies, leveraging grounding models to parse complex UIs (\eg, UGround~\cite{gou2025navigating}, Aria-UI~\cite{ariaui}, UITARS~\cite{qin2025ui}, OmniParser~\cite{wan2024omniparser}, Ferret-UI~\cite{you2024ferret}). While these systems show the promise of LLMs/VLMs for GUI automation, they largely target explicit task completion with dense, well-defined rewards. By contrast, mobile ad detection involves concealed objectives, limited observability, and sparse triggers. Our work instead targets ad-relevant states, facilitating discovery of hidden advertising behaviors and expanding agentic UI reasoning.

%% file: chapter/discussion.tex
\section{Discussion}
\blue
{\textbf{Scalability and deployment.}
\systemname leverages the general reasoning capabilities of LLMs and VLMs to achieve more effective ad detection, while incurring additional computation and latency. To mitigate this overhead, \systemname adopts selective VLM invocation and is instantiated with relatively compact models such as GPT-4o-mini, which opens up opportunities for future on device deployment to enhance privacy or for adopting more lightweight models to balance detection performance and latency in practical settings. For large scale analysis, mobile SoC clusters can be employed to enable parallel execution and improve throughput.}

\noindent \blue{\textbf{Robustness under practical constraints.}
\systemname operates under practical constraints shared by current LLM based systems. In scenarios involving CAPTCHAs or other human verification mechanisms, automated interaction may be impeded. Addressing such cases represents a broader challenge for LLM driven task automation, where systems must reason and act under explicitly restricted interaction channels. Exploring more robust strategies for handling these constraints remains an important direction for future research.}

%% file: chapter/conclusion.tex
\section{Conclusion}
In this paper, we present \systemname, an agentic multimodal framework that unifies static, visual, temporal, and experiential signals to enable robust ad detection in modern mobile apps. Experiments show that it outperforms existing approaches in both coverage and efficiency, while generalizing across apps through reusable trajectories. Beyond mobile advertising, our design points to a broader paradigm of multimodal reasoning–guided UI navigation for securing and testing mobile ecosystems, highlighting its potential as a versatile foundation for future research.

%% file: chapter/acknowledgement.tex
\section{Acknowledgment}
This research was supported in part by the National Natural Science Foundation of China under Grant No. 62432004, and by a grant from the Guoqiang Institute, Tsinghua University.

%% file: chapter/appendix.tex
\appendix

\section{Workflow}
\label{sec:workflow}
Before initiating online navigation, we perform an \emph{offline profiling} phase that consists of \emph{static analysis} and \emph{dynamic probing}.  
\subsection{Offline Profiling: Static Analysis}
The static analysis pipeline constructs three types of priors:
\emph{screen priors}, \emph{slot priors}, and \emph{trigger priors}.
Given an APK and its decompiled resources as input, we leverage Androguard and related toolchains to extract three types of ad-related priors.

Specifically, (i) \textbf{Screen prior}: we parse the manifest to collect all registered Activities and extract declared permissions and metadata, which are matched against known ad library prefixes.
\begin{lstlisting}[label={lst:screen_prior}]
parse manifest from {decompiledRes}, extract permissions and metadata
if permissions/metadata matches {config.sdk_prefixes}:
    record as {offlineResult.screen}
\end{lstlisting}
(ii) \textbf{Slot prior}: we traverse layout resource files and match View class names against a preconfigured set of ad SDK signatures to localize potential ad widgets. Using the compiled resource mapping table, string identifiers are further resolved into unique hexadecimal resource IDs.
\begin{lstlisting}[label={lst:slot_prior}]
iterate over each {layoutFile} in {decompiledRes}:
    parse XML tree as {viewTree}
    iterate over {node} in {viewTree}:
        if {node.class} matches {config.sdk_prefixes}:
            resolve resource ID using resource map, as {id_hex}
            record {node.class}, {id_hex} into {offlineResult.slot}
\end{lstlisting}
(iii) \textbf{Trigger prior}: we scan bytecode across all classes, and upon detecting ad SDK API invocations or listener callbacks, it recursively backtraces along the inheritance chain. Once a superclass matches a manifest-registered Activity, the corresponding code evidence is attributed to that Activity.  
This process transforms isolated code fragments into an Activity-centric contextual mapping.
\begin{lstlisting}[label={lst:trigger_prior}]
iterate over each {class} in {apkFile}:
    if {class} calls Ad APIs or implements Ad Listeners:
        extract method signatures as {clues}
        # Backtrace to find the owner Activity
        while {currClass} has superclass:
            if {currClass} is an Activity:
                associate {clues} with {currClass} in {offlineResult.trigger}
                break loop
            set {currClass} to {currClass.super}
\end{lstlisting}

\subsection{Offline Profiling: Dynamic Probing}
Following static analysis, we perform \emph{dynamic probing} by executing the app with DroidBot to construct a coarse UTG and collect system logs. Each UI transition is recorded as a timestamped event sequence. 
\begin{lstlisting}[label={lst:dynamic}]
run DroidBot on {apkFile}, record transition graph as {UTG} and system logs as {sysLogs}
\end{lstlisting}
Next, we extract ad-related network traffic by filtering log entries using keyword matching and known ad domains, yielding a set of timestamped requests.  
To derive \textbf{network prior}, we temporally correlate UI events with network requests within a sliding window $\Delta$, and associate triggering events with nearby ad traffic. 
\begin{lstlisting}[label={lst:network}]
iterate over {log} in {sysLogs}:
    if {log} matches keywords or known ad domains:
        extract {log.url} and {log.timestamp}, add to {adTraffic}
sort {UTG.events} and {adTraffic} by time
iterate over {event} in {UTG.events}:
    find {traffic} in {adTraffic} where time difference < $\Delta$
    if {traffic} exists:
        link {event} to {traffic} in {networkPrior}
\end{lstlisting}
Finally, we merge these network-level priors with the static priors to produce a unified offline knowledge base, which serves as structured guidance for subsequent online navigation.

\subsection{Multi-Source Prompt Construction}
\label{sec:prompt}
Based on the integrated offline knowledge, we further construct a \emph{multi-source structured prompt} to guide per-step decision making.  
At each state, the prompt aggregates four complementary information sources: \\(i) \textbf{current screen options}, which enumerate all actionable views with textual or visual descriptions; 
\begin{lstlisting}[label={lst:prompt-screen}]
iterate over {view} in {views}:
    combine {view.text}, {view.desc}, {view.vlm_desc}
    add to {prompt.screen}
\end{lstlisting}
(ii) \textbf{offline priors}, injecting activity-level, component-level, and method-level ad evidence derived from offline profiling; 
\begin{lstlisting}[label={lst:prompt-offline},columns=flexible,keepspaces=true,showstringspaces=false]
if {currentState.activity} in {priors.ad_activities}:
    add "Activity Listed as ad-related" to {prompt.offline}
if any {view} in {views} matches {priors.ad_components}:
    add "{view} is a potential Ad View" to {prompt.offline}
if {currentState.activity} in {priors.methods_by_activity}:
    get methods {m} from {priors.methods_by_activity}
    add "{m} are identified as Ad Methods" to {prompt.offline}
}
\end{lstlisting}
(iii) \textbf{strategic context}, capturing the local UTG neighborhood and recent interaction history to avoid redundant exploration; 
\begin{lstlisting}[label={lst:prompt-context},columns=flexible,keepspaces=true,showstringspaces=false]
get neighborhood nodes of {currentState} from {UTG} within 2 hops as {localNeighbors}
iterate over {node} in {localNeighbors}:
    add {node.id}, {node.visits}, {node.score} to {prompt.context}
get recent {k} steps from {history} as {recentPath}
add {recentPath} to {prompt.context}
\end{lstlisting}
and (iv) \textbf{past experiences}, retrieved from an experience database to reuse previously successful behaviors. 
\begin{lstlisting}[label={lst:prompt-exp}]
retrieve relevant experiences from {experienceDB} similar to {currentState} as {rel_exp}
iterate over {exp_item} in {rel_exp}:
    add {exp_item} to {exp}
\end{lstlisting}
These elements are concatenated into a unified, structured prompt, providing the LLM with a holistic and temporally grounded view of the current decision context.

\subsection{LLM-driven UI navigation}
Building upon the structured prompt, we employ an \emph{LLM-driven UI navigation policy} to actively trigger ad-related behaviors.  
At each step, the agent observes the current device state and first checks a success condition, where reaching a known ad-related Activity terminates the episode and records the corresponding trigger path. 
The discovered path is summarized into a reusable heuristic and stored in an experience database to guide future exploration.  
\begin{lstlisting}[label={lst:navigation-success}]
set {step} to 0
initialize {currentTrajectory} as empty
while {step} < {maxSteps}:
    get device state as {state}
    # 1. Check Success
    if {state.activity} matches {offlinePriors.success_activities}:
        summarize {currentTrajectory} as {heuristic}
        save {heuristic} to {experienceDB}
        restart {app}
        continue loop
\end{lstlisting}
If the app enters an abnormal state, a lightweight recovery procedure is applied to ensure robustness.
\begin{lstlisting}[label={lst:navigation-recovery}]
while {step} < {maxSteps}:
    # 1. ... 2. Recovery Mechanism
    if {app} is crashed or in background:
        perform "Back" or "Restart"
        continue loop
\end{lstlisting}
The agent then extracts all actionable UI elements, optionally enriching them with visual descriptions via a VLM. 
\begin{lstlisting}[label={lst:navigation-ui}]
while {step} < {maxSteps}:
    # 1,2. ... 3. Perception
    if app is hierarchy-apps:
        get views from Accessibility Service as {views}
    else if app is canvas-apps:
        capture screenshot
        run vision detector to segment {views} regions
        run VLM to annotate {views}
\end{lstlisting} 

Conditioned on the multi-source prompt constructed in Section~\ref{sec:prompt}, the LLM selects the next action and estimates its ad relevance. 
\begin{lstlisting}[label={lst:navigation-select}]
while {step} < {maxSteps}:
    # 1,2,3. ... # 4. Reasoning and Action
    execute function "ConstructPrompt" with {state}, {views}, {offlinePriors}... 
        get result as {prompt}
    query LLM with {prompt}, get {actionIdx} and {adScore}
    execute {action} on {device}
\end{lstlisting}
Finally, the UTG and interaction history are updated, and the selected action is executed, enabling iterative, experience-aware navigation until the step budget is exhausted.
\begin{lstlisting}[label={lst:navigation-update}]
while {step} < {maxSteps}:
    # 1,2,3,4. ... # 5. Update State
    update {history} with {state}, {action}
    if node, edge not in UTG:
        update node, edge in {UTG}
    update node score in {UTG} with {adScore}
    increment {step}
\end{lstlisting}

\section{Examples}
\label{sec:examples}
\subsection{Hierarchy-based apps w/o VLM}

\begin{figure}[!h]
\includegraphics[width=0.49\textwidth]{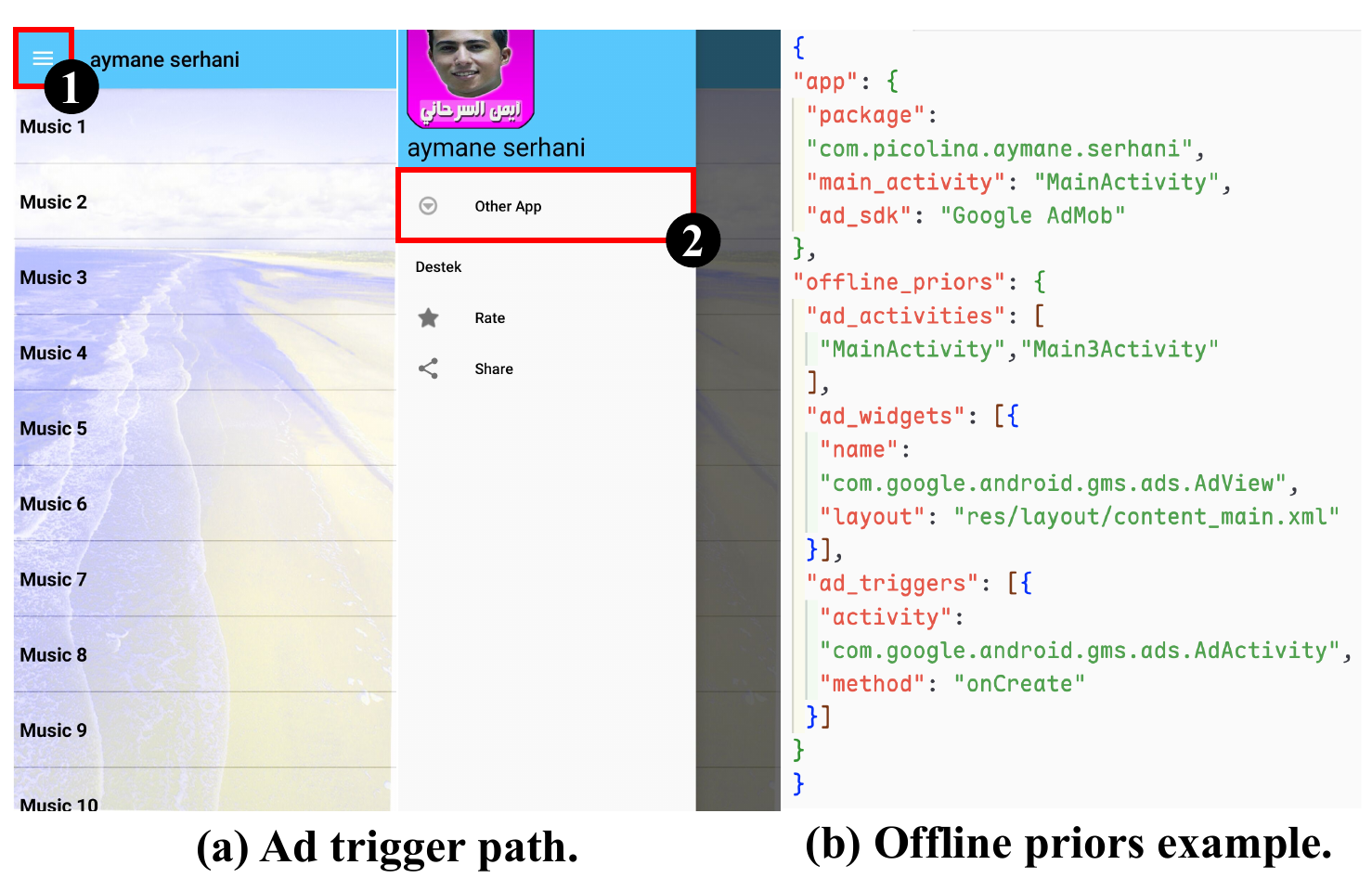}
  \vspace{-0.6cm}
  \caption{An example of a hierarchy-based app.}
  \label{fig:appendix-exp1}
\end{figure}

We first present a simple example of a \emph{hierarchy-based} app that does not require VLM assistance.  
Figure~\ref{fig:appendix-exp1}(a) illustrates the ad-triggering workflow in this app: \systemname first clicks (1) the navigation bar icon to open the side drawer, and then clicks (2) \emph{Other Apps}, which redirects to the Google Play Store, which is a classic ad-triggering path.  
Figure~\ref{fig:appendix-exp1}(b) shows a JSON example of the offline profiling results.

The detailed online navigation workflow of \systemname in this case is described as follows:
At the entry screen, \systemname constructs the prompt as shown below.  
Specifically, \systemname derives the \emph{current screen options} from the view tree provided by the Android Accessibility Service, and builds the \emph{strategic context} using the UTG and interaction history.  
Since this is the initial state, the interaction history only contains the app launch event.  
Finally, relevant past experiences are retrieved, where the first entry highlights navigation elements that exactly match the ad-triggering path in this example.
For clarity, we present the first prompt with a relatively complete information structure, while subsequent prompt examples only include the most critical elements involved in decision making.

\begin{lstlisting}[style=prompt, label={lst:prompt}, escapeinside={(*}{*)}
]
(*\textbf{[System prompt.]}*) You are an agent designed to ...
(*\textbf{[Integrated prompt.]}*)
(*\underline{1. Current Screen Options}*)
- View 0: Type='ImageButton', Text='Open navigation drawer'
- View 1: Type='ListView', Res-ID='com.picolina.aymane.serhani:id/list1'
- View 2: Type='LinearLayout', Text='Music 1'
- ...
- View 11: Type='LinearLayout', Text='Music 10'
- View 12: Type='BackButton', Text='[BACK] Return to previous screen'
(*\underline{2. Static App Knowledge}*)
[Activity Match] Current activity 'MainActivity' is listed as ad-related.
[Activity Match] This activity contains potential ad trigger(s) in method(s): ['onCreate'].
[Component Match] A component with resource_id 'com.picolina.aymane.serhani:id/adView' is a known ad container.
[General Info] App uses ad libraries: ['Google AdMob']
(*\underline{3. Strategic Context}*)
(a) Annotated Local Map (from UTG, 2-hop neighborhood)
Current State[4662c1] (visited: 1 times), ad_score: 0.05
**Reachable in 1-hop(s):**
- State: [38961e], event: 'KeyEvent(state=4662c1, name=BACK)', ad_score: 0.10 (visited: 0 times)
- State: [09aff5], event: 'TouchEvent(state=4662c1, view=[button alt='Open navigation drawer' bound_box=0,102,168,270][/button])', ad_score: 0.10 (visited: 0 times)
- State: [b1a47a], event: 'TouchEvent(state=4662c1, view=[0,270,1080,480-LinearLayout-])', ad_score: 0.10 (visited: 0 times)
- ...
**Reachable in 2-hop(s):**
- State: [5c7efb], event: 'TouchEvent(state=79130a, view=[389,622,691,718-ImageButton-])', ad_score: 0.10 (visited: 0 times)
- State: [e3f97e], event: 'KeyEvent(state=401904, name=BACK)', ad_score: 0.10 (visited: 0 times)
(b) Recent History
- Step 1 [166c5c] -> [4662c1] RestartAppEvent()
(*\underline{4. Past Experiences}*)
- Interacting with music playback or navigation elements often triggers promotional ads in media apps.
- Interacting with promotional banners or offers related to specific age groups or financial products often triggers advertisements.
- Interacting with buttons or views that suggest accessing additional content or resources often triggers advertisement offers.
\end{lstlisting}
Given this prompt, the LLM produces the following action selection and corresponding reasoning:
\begin{lstlisting}[style=prompt, label={lst:prompt}, escapeinside={(*}{*)}
]
{"reasoning": "Choosing element 0, as the 'Open navigation 
drawer' button may lead to additional content or promotional 
offers, which could include advertisements.",
"ad_score": 0.05,
"choice": 0}"
\end{lstlisting}
By clicking the navigation button, the app opens the side drawer and transitions to a state that contains ad-related buttons.  
At this point, we query the LLM again with the following prompt, where only the most salient information is shown for brevity.
\begin{lstlisting}[style=prompt, label={lst:prompt}, escapeinside={(*}{*)}
]
(*\textbf{[System prompt.]}*) You are an agent designed to ...
(*\textbf{[Integrated prompt.]}*)
(*\underline{1. Current Screen Options}*)
- View 0: Type='LinearLayoutCompat', Text='N/A'
- View 1: Type='CheckedTextView', Text='Other App'
- View 2: Type='LinearLayoutCompat', Text='Rate'
- ...
(*\underline{2. Static App Knowledge}*) ...
(*\underline{3. Strategic Context}*)
(a) Annotated Local Map (from UTG, 2-hop neighborhood) ...
(b) Recent History
- Step 1 [166c5c] -> [4662c1] RestartAppEvent()
- Step 2 [4662c1] -> [09aff5] TouchEvent(state=4662c1, view=[button alt='Open ...ation drawer' bound_box=0,102,168,270][/button])
(*\underline{4. Past Experiences}*)
- Interacting with buttons or views that suggest accessing additional content or resources often triggers advertisement offers.
- ...
\end{lstlisting}
The LLM decides to click the \emph{Other App} button, which triggers an app-promotion transition.
\begin{lstlisting}[style=prompt, label={lst:prompt}, escapeinside={(*}{*)}
]
{"reasoning": "Choosing element 1, as the text 'Other App' suggests a potential for ad-related content or promotions, which aligns with the heuristics for identifying ad triggers.",
"ad_score": 0.05,
"choice": 1}"
\end{lstlisting}
After successfully triggering the ad, \systemname prompts the LLM to summarize the interaction history and distill it into a corresponding experience.
\begin{lstlisting}[style=prompt, label={lst:prompt}, escapeinside={(*}{*)}
]
You are an expert Android app tester specializing in identifying ad-triggering patterns. Your task is to ...
Step 1: Touched a 'ImageButton' with text/desc: 'Open navigation drawer'.
Step 2: Touched a 'CheckedTextView' with text/desc: 'Other App'.
\end{lstlisting}
In this case, the distilled experience is shown below, which is consistent with our intuition.
\begin{lstlisting}[style=prompt, label={lst:prompt}, escapeinside={(*}{*)}
]
Interacting with navigation options that lead to external app suggestions often triggers advertisement displays.
\end{lstlisting}

\subsection{Hierarchy-based apps with VLM}
\begin{figure}[!h]
\includegraphics[width=0.49\textwidth]{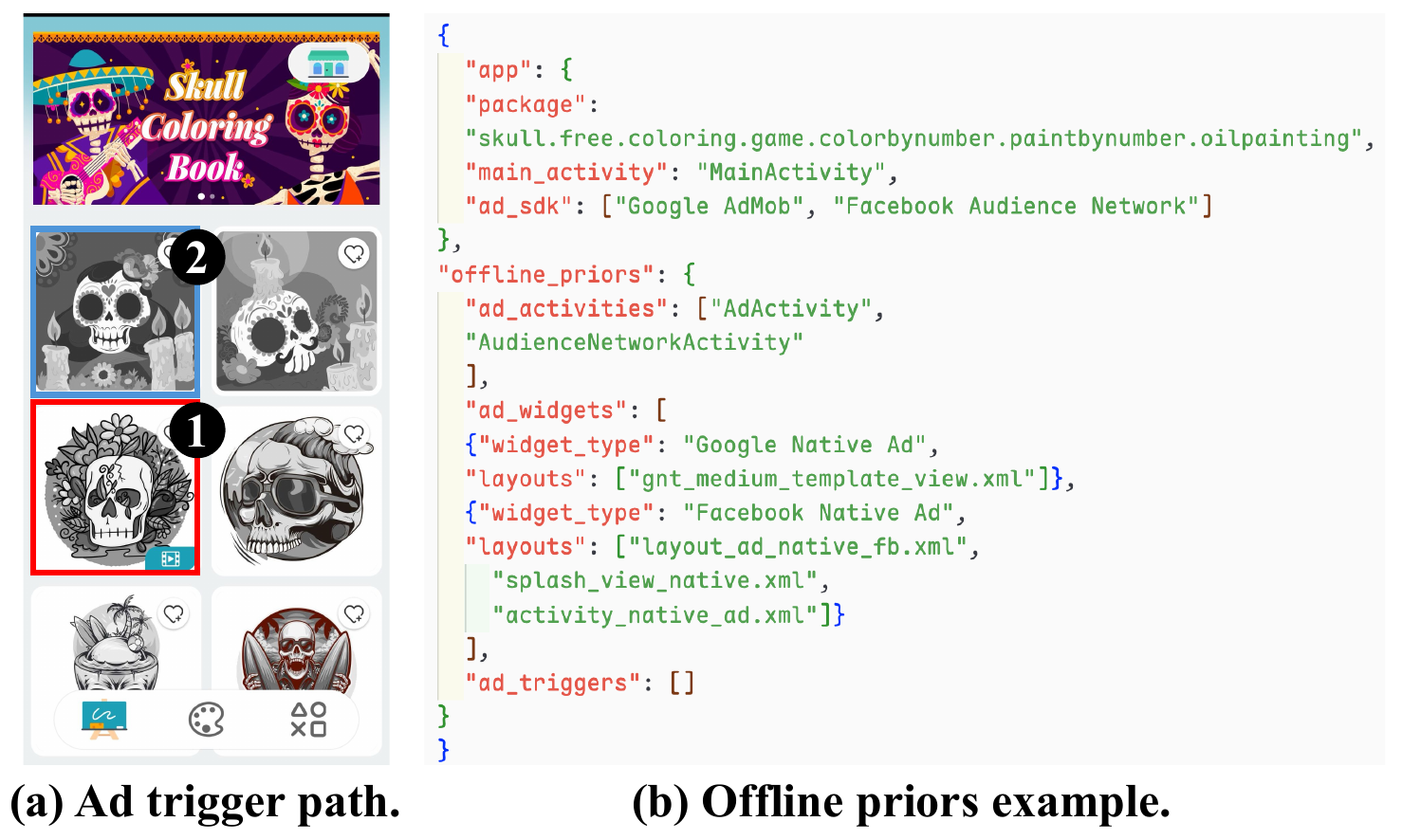}
  \vspace{-0.6cm}
  \caption{An example of a hierarchy-based app with VLM-annotated UI element.}
  \label{fig:appendix-exp2}
\end{figure}
This example illustrates a case where \systemname leverages a VLM to enrich UI elements with semantic descriptions.
Figure~\ref{fig:appendix-exp2} shows the detailed scenario. At the entry screen, there already exists an element (\ie, \circlednum{1}) that triggers an advertisement video.
However, these \texttt{ImageView} elements (\ie, \circlednum{1} and \circlednum{2}) lack content descriptions in the view tree, making it difficult to distinguish among different image components.
Therefore, \systemname first employs a VLM to annotate cropped screenshots of UI elements with semantic labels.
\begin{lstlisting}[style=prompt, label={lst:prompt}, escapeinside={(*}{*)}]
You are an expert mobile ad detector. Analyze this UI
screenshot and label the element according to the following
categories: 
[AD]: ..., [POTENTIAL_AD]: ..., [UI_ELEMENT]: ...
\end{lstlisting}
For these \texttt{ImageView} elements, the VLM produces the following semantic annotations.
\begin{lstlisting}[style=prompt, label={lst:prompt}, escapeinside={(*}{*)}]
# For element #1:
[AD] A button with a play icon likely for a video advertisement.
# For element #2:
[UI_ELEMENT] A decorative skull and candles, likely part of the app's theme.
\end{lstlisting}
Based on these semantic annotations, we construct the following query prompt for the LLM.
\begin{lstlisting}[style=prompt, label={lst:prompt}, escapeinside={(*}{*)}]
(*\textbf{[System prompt.]}*) You are an agent designed to ...
(*\textbf{[Integrated prompt.]}*)
(*\underline{1. Current Screen Options}*)
- View 0: Type='ViewPager', Text='More You have no works yet, go to color your paintings!'
- View 1: Type='ViewPager', Text='N/A', Res-ID='None'
- View 2: Type='RecyclerView', Text='N/A', Res-ID='None'
- View 3: Type='ImageView', Text='[POTENTIAL_AD] A speech bubble with a storefront icon, possibly indicating a virtual store or in-app purchase option.'
- ...
- View 8: Type='ImageView', Text='[UI_ELEMENT] A decorative skull and candles, likely part of the app's theme.'
- View 9: Type='ImageView', Text='[AD] A button with a play icon likely for a video advertisement.'
- ...
(*\underline{2. Static App Knowledge}*) ...
(*\underline{3. Strategic Context}*)
(a) Annotated Local Map (from UTG, 2-hop neighborhood) ...
(b) Recent History ...
(*\underline{4. Past Experiences}*) ...
\end{lstlisting}
Based on the enriched semantic descriptions of UI elements, the LLM selects element~\circlednum{1} for interaction, which triggers an interstitial video advertisement.
\begin{lstlisting}[style=prompt, label={lst:prompt}, escapeinside={(*}{*)}
]
{"reasoning": "Choosing element 9, as it is a clear ad indicator with a play icon, which is a high-confidence ad trigger.",
"ad_score": 0.8,
"choice": 9}"
\end{lstlisting}
\subsection{Canvas-based apps}
\begin{figure}[!h]
\includegraphics[width=0.49\textwidth]{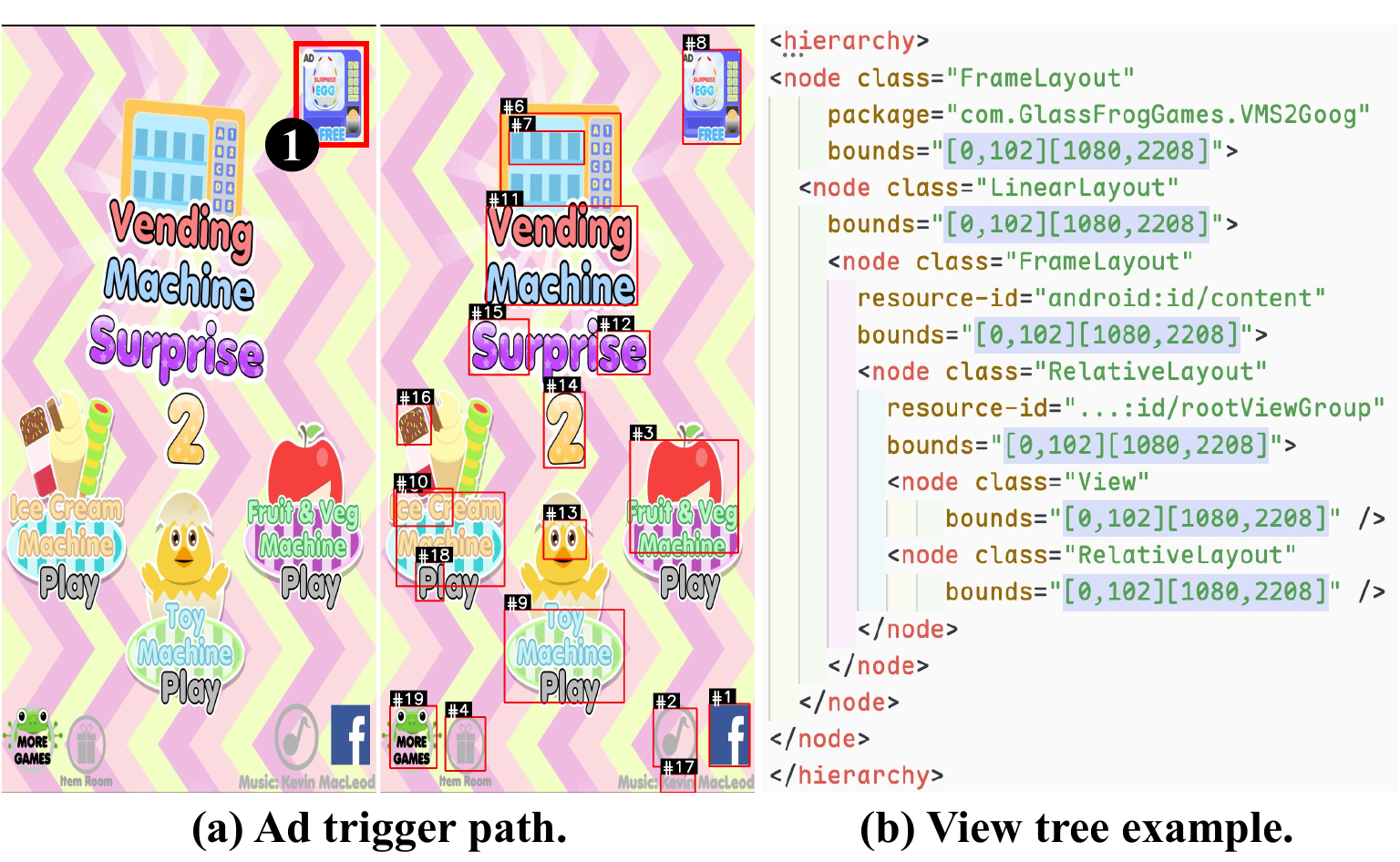}
  \vspace{-0.6cm}
  \caption{An example of a canvas-based app.}
  \label{fig:appendix-exp3}
\end{figure}

This example illustrates a canvas-rendered case. As shown in Figure ~\ref{fig:appendix-exp3}(b), the view tree obtained via the Android Accessibility Service contains only the bounding box of the entire canvas, without fine-grained semantic information for individual UI elements. 
\systemname first applies a hybrid vision detector to identify coarse-grained UI component regions, as illustrated in Figure ~\ref{fig:appendix-exp3}(a). 
Subsequently, a VLM is leveraged to generate semantic annotations for each detected component.
\begin{lstlisting}[style=prompt, label={lst:prompt}, escapeinside={(*}{*)}]
You are an expert Android UI analyst specializing in ad detection. Your task is to ...
Your output MUST be a valid JSON array, where each object contains the 'id' (the number from the box) and 'description' (your analysis, with annotations if necessary) ...
\end{lstlisting}
For Figure~\ref{fig:appendix-exp3}(a), the VLM provides the following annotation.
\begin{lstlisting}[style=prompt, label={lst:prompt}, escapeinside={(*}{*)}]
[{"id": 1, "description": "A blue Facebook icon indicating social media sharing."}, ...
{"id": 8, "description": "[AD HINT] A large promotional image featuring a 'FREE' surprise egg."}, ...
{"id": 19, "description": "[AD HINT] A round button featuring a frog and the label 'More Games'."}]"
\end{lstlisting}
Based on the above VLM annotation, we construct a corresponding current screen options prompt that incorporates semantic information.
\begin{lstlisting}[style=prompt, label={lst:prompt}, escapeinside={(*}{*)}]
(*\textbf{[System prompt.]}*) You are an agent designed to ...
(*\textbf{[Integrated prompt.]}*)
(*\underline{1. Current Screen Options}*)
- View 0, Text='A blue Facebook icon indicating social media sharing.'
- ...
- View 8, Text='[AD HINT] A large promotional image featuring a 'FREE' surprise egg.'
- ...
(*\underline{2. Static App Knowledge}*) ...
(*\underline{3. Strategic Context}*)
(a) Annotated Local Map (from UTG, 2-hop neighborhood) ...
(b) Recent History ...
(*\underline{4. Past Experiences}*) ...
\end{lstlisting}
Based on these key cues, \systemname successfully selects the icon that triggers an advertisement.
\begin{lstlisting}[style=prompt, label={lst:prompt}, escapeinside={(*}{*)}]
{"reasoning": "Choosing element 7, as it features a large promotional image with the text 'FREE', which is a high-confidence ad trigger indicating a potential advertisement opportunity.",
"ad_score": 0.8,
"choice": 7}"
\end{lstlisting}

\section{Mathematical Formalization and Coverage Analysis}
\label{app:mathematical}

We formulate mobile ad detection as a sequential decision-making problem under partial observability, where an autonomous explorer interacts with a target application to uncover ad-related UI states.

\subsection{Environment and Interaction}
An Android application is modeled as an unknown UI transition system
$\mathcal{E} = (\mathcal{S}, \mathcal{A}, \mathcal{T})$,
where $\mathcal{S}$ denotes the set of latent UI states (activities or screens),
$\mathcal{A}$ denotes the set of executable UI actions (e.g., click, back),
and $\mathcal{T}: \mathcal{S} \times \mathcal{A} \rightarrow \mathcal{S}$ is the state transition function. A target \emph{ad trigger instance} is a tuple $\alpha = (s, \tau)$, where $s \in \mathcal{S}$ is a latent state in which an ad is exposed, $\tau$ denotes the required interaction context required to trigger the ad.
Let $\mathcal{A}_{\mathrm{ads}}$ denote the set of all such instances.

\textbf{Augmented State Representation (\systemname).}
Standard baselines (e.g., DroidBot) typically rely on a restricted observation space, such as the UI view hierarchy $M_t$. In contrast, MANA operates on an \textit{Augmented Context} $Z_t$, integrating heterogeneous signals to resolve ambiguity:
\begin{equation}
    Z_t = \Phi(M_t, V_t, \Sigma, H_t, \mathcal{E}_{mem})
\end{equation}
where:
\begin{itemize}
    \item $M_t$: Structural metadata from Android Accessibility Service (widget hierarchy, text).
    \item $V_t$: Optional visual features (e.g., screenshots processed by VLM) for canvas-rendered UIs.
    \item $\Sigma$: Static analysis priors derived from offline profiling (Screen, Slot, and Trigger priors).
    \item $H_t$: The interaction history sequence, providing temporal grounding.
    \item $\mathcal{E}_{mem}$: Retrieved cross-app experiences from the memory bank.
\end{itemize}
This enriched context is the primary input to the LLM-based policy $\pi_{\text{MANA}}$.

\subsection{Interaction History and UI Transition Graph}
\textbf{Interaction History.}
\systemname maintains a bounded interaction history
\[
H_t = \{(o_i, a_i)\}_{i = t - L_t}^{t-1}.
\]
where $K_t$ is a dynamic window size that adapts based on recent state diversity. Specifically, let $\mathcal{U}_t$ denote the set of unique UI states observed in the most recent $K_\text{base}$ interactions. Then the adaptive window size is defined as
\[
K_t =
\begin{cases}
\lceil 1.5 \cdot K_\text{base} \rceil, & \text{if } |\mathcal{U}_t| \le 2, \\
K_\text{base}, & \text{otherwise}.
\end{cases}
\]
The history buffer stores tuples of the form
\[
(o_i, a_i, o_{i+1}, r_i, \tau_i),
\]
where $r_i$ indicates whether an ad was observed during this transition, and $\tau_i$ denotes the timestamp. The interaction history $H_t$ serves as the temporal grounding context, guiding \systemname to avoid loops or repeated exploration of the same UI states.

\noindent\textbf{UI Transition Graph (UTG).}
\systemname incrementally constructs a UI Transition Graph
$\mathcal{G}_t = (V_t, E_t)$,
where each node $v \in V_t$ corresponds to an UI state,
and each directed edge $(v, a, v') \in E_t$ represents an observed transition
triggered by action $a$.
Each node $v$ stores:
(i) structural metadata (activity name, widget hierarchy),
(ii) semantic summaries inferred by the language model,
and (iii) an accumulated ad-relevance belief score $S(v)$.

Upon observing a transition $(o_t, a_t, o_{t+1})$,
the UTG is updated as:
\[
V_{t+1} = V_t \cup \{v_{t+1}\}, \quad
E_{t+1} = E_t \cup \{(v_t, a_t, v_{t+1})\}.
\]
\noindent\textbf{Ad-Relevance Scoring.}
At each step, the language model estimates an instantaneous ad relevance
$\hat{s}_t \in [0,1]$ based on $(o_t, H_t, \mathcal{G}_t)$.
The belief score of the current node is updated via an exponential moving average:
\[
S_t = (1-\alpha) S_{t-1} + \alpha \hat{s}_t,
\]
where $\alpha \in (0,1)$ controls temporal smoothing.

\subsection{The Coverage Objective and the Challenge of Aliasing}
\systemnames goal is to maximize the discovery of distinct ad trigger instances within a finite interaction budget $T$. We define the \emph{ad coverage} of a policy $\pi$ as the set:
\[
\mathcal{C}_{\text{ad}}(\pi, T) = \{ \alpha \in \mathcal{A}_{\text{ads}} \mid \alpha \text{ is triggered by } \pi \text{ within } T \}.
\]
The objective is implicitly to maximize $\mathbb{E}[|\mathcal{C}_{\text{ad}}(\pi, T)|]$.

Standard baseline policies (e.g., $\pi_{\text{base}}$) that rely solely on the structural metadata $M_t$ suffer from \textbf{Structural Aliasing}. We define the \emph{Structural Equivalence Class} for a latent state $s$ as:
\[
[s]_{\text{struct}} = \{ s' \in \mathcal{S} \mid M(s') = M(s) \}.
\]
In content-rich apps, distinct latent states often share identical UI structures (e.g., distinct pages in a dictionary sharing the same layout). A policy $\pi_{\text{base}}(a \mid M_t)$ cannot distinguish between them, leading to \textbf{probabilistic loops}: the agent may perpetually interact within the same equivalence class without progressing to deeper, unseen states or ad-triggering logic (e.g., cyclic navigation by repeatedly clicking the \emph{next word} button in a dictionary app). This directly limits achievable coverage.

\subsection{\systemname: Implicit Decision Criterion}
To reason about how MANA resolves aliasing and achieves broader coverage,
we introduce an \emph{implicit decision criterion} that analytically
characterizes the preference underlying the LLM's action selection.
Concretely, we define
\begin{equation}
    J(\pi) \doteq
    \mathbb{E} \left[
    \sum_{t=0}^{T}
    \left(
    R_{\text{sem}}(Z_t)
    - \lambda \cdot N(v_{next})
    \right)
    \right],
\end{equation}
which serves as an analytical characterization of the trade-off guiding the
LLM's per-step decisions.

Here, the criterion balances two complementary signals derived from the augmented
context $Z_t$:
\begin{itemize}
    \item \textbf{Semantic Gain ($R_{\text{sem}}$):} An estimated utility indicating
    whether an action is likely to lead to ad-related logic, inferred from textual
    and visual semantics in $Z_t$ (e.g., prioritizing semantically meaningful navigation
    over repetitive interactions).
    \item \textbf{Exploration Penalty ($N$):} A penalty derived from visitation counts
    in the UI Transition Graph (UTG), discouraging redundant traversals of structurally
    identical states.
\end{itemize}

\subsection{Theorem: Coverage Dominance}
\begin{lemma}[Loop Escape Implies Coverage Expansion]
Escaping a structural loop that is absorbing under $\pi_{base}$
strictly increases the set of reachable latent states and ad triggers
under a finite interaction budget.
\end{lemma}
\begin{proof}
Since $\pi_{base}$ remains confined to $\mathcal{S}_{loop}$ with non-zero probability mass, any policy that exits $\mathcal{S}_{loop}$ expands the reachable set.
\end{proof}
\begin{theorem}[Resolution of Structural Loops]
Let $\mathcal{C}_{ad}(\pi)$ denote the set of ad trigger instances
$\alpha = (s,\tau,c)$ covered by policy $\pi$ within budget $T$.
Under the assumption that ad triggers are reachable via semantically distinguishable transitions, MANA ensures broader coverage than structural baselines:
\begin{equation}
    \mathbb{E}[|\mathcal{C}_{ad}(\pi_{\text{MANA}})|] > \mathbb{E}[|\mathcal{C}_{ad}(\pi_{base})|]
\end{equation}
\end{theorem}

\begin{proof}
Since each ad trigger requires reaching a specific latent state under a compatible interaction context, increased escape from structural loops strictly enlarges the reachable ad-trigger space.
Consider a scenario where the agent is trapped in a structural loop $\mathcal{S}_{loop} \subset [s]_{struct}$ (e.g., infinite scrolling).

\textbf{1. Baseline Stagnation:} Since $\pi_{base}$ relies on invariant structure $M$, the probability of selecting the exit action $a_{exit}$ remains stationary and low. The agent performs a random walk within the equivalence class.

\textbf{2. MANA Escape Mechanism:} MANA utilizes $Z_t$ to break the loop via two complementary mechanisms:
\begin{itemize}
    \item \textbf{Semantic Un-aliasing:} The LLM discerns semantic differences in $Z_t$ even when $M$ is identical (e.g., "Home" vs. "Next"). It assigns a higher prior to the exit action: $R_{\text{sem}}(a_{exit}) > R_{\text{sem}}(a_{loop})$, enabling an immediate semantic breakout.
    \item \textbf{Probabilistic Escape:} If semantic cues are ambiguous, the visitation count $N(\mathcal{S}_{loop})$ increases monotonically. The penalty term $-\lambda N$ grows, progressively reducing the utility of the loop. Mathematically, there exists a time $t$ where $\mathbb{E}[R_{exit} - \lambda N_{low}] > \mathbb{E}[R_{loop} - \lambda N_{high}]$, forcing a policy shift.
\end{itemize}
Thus, MANA escapes local optima that trap baselines, ensuring a superset of state coverage.
Notably, the escape mechanism does not rely on perfect semantic inference:
even when $R_{\text{sem}}$ is noisy, the monotonic growth of $N(\cdot)$
ensures eventual exploration pressure.
\end{proof}